\documentclass[12pt, draftclsnofoot, onecolumn]{IEEEtran}
\usepackage{setspace}
\usepackage{hyperref}
\usepackage{amsthm}
\usepackage{epsfig}
\usepackage{amscd}
\usepackage{float,color,graphicx,subfigure}
\usepackage{amssymb}
\usepackage{amsmath}
\usepackage{lscape}
\usepackage{color}
\usepackage{tabularx}
\usepackage{multirow}
\usepackage{booktabs}
\usepackage{siunitx}
\usepackage{breqn}
\usepackage{lipsum}
\usepackage{mathrsfs}
\usepackage{placeins}
\newcommand{\mb}[1]{\boldsymbol{#1}}

\newtheorem{prop}{Proposition}
\newtheorem{example}{Example}

\newcommand{\thickhline}{\noalign{\hrule height 0.8pt}}
\normalsize
\begin{document}
\sloppy
\title{Novel Method for Multi-Dimensional Mapping of Higher Order Modulations for BICM-ID Over Rayleigh Fading Channels}

\author{
\IEEEauthorblockN{Hassan M. Navazi and Md. Jahangir Hossain, \emph{Member, IEEE}\\}
\IEEEauthorblockA{
The University of British Columbia, Kelowna, BC, Canada\\
\emph{hnavazi@alumni.ubc.ca}, \emph{jahangir.hossain@ubc.ca}
}
}
\maketitle
\begin{abstract}
 Multi-dimensional (MD) mapping  offers more flexibility in mapping design for bit-interleaved coded modulation with iterative decoding (BICM-ID) and potentially improves the bandwidth efficiency. However,  for higher order signal constellations, finding suitable MD mappings is a very complicated task due to the large  number of  possible mappings. In this paper, a novel mapping  method is introduced to construct efficient MD mappings to improve the error performance of BICM-ID over Rayleigh fading
 channels. We propose to break the MD mapping design problem into four distinct $2$-D mapping functions.  The $2$-D mappings are designed such that the resulting MD mapping improves the BICM-ID error performance at low signal to noise ratios (SNRs). We also develop cost functions that can be optimized to improve the  error performance at high SNRs. \color{black} The proposed mapping method is very simple compared to well-known mapping methods, and it can achieve suitable MD mappings for different modulations including higher order modulations for BICM-ID. Simulation results show that our mappings significantly outperform the previously known mappings at a target bit error rate (BER) of $10^{-6}$. Our mappings also offer a lower error-floor compared to their well-known counterparts.
\end{abstract}

\begin{IEEEkeywords}
BICM-ID, multi-dimensional signal mapping, QAM, Rayleigh fading channels.
\end{IEEEkeywords}

\section{Introduction}
\label{intro}
\color{black} Trellis coded modulation (TCM) \cite{TCM}  improved the bit error rate (BER) performance of coded modulation by maximizing the Euclidean distance among the coded signal sequences.  To improve the BER of TCM over Rayleigh fading channels, bit interleaved coded modulation (BICM)  was introduced by Zehavi \cite{Zehavi}.   \color{black} BICM offers good performance over Rayleigh fading channels. However, the random modulation caused by the interleaver degrades the BICM performance over additive white Gaussian noise (AWGN) channels. To address this problem, iterative decoding was used  at the receiver. The resulting system is referred to as BICM with iterative decoding (BICM-ID) and is investigated in  \cite{BICMID}-\cite{Benedetto}.   
 BICM-ID offers better performance over AWGN and Rayleigh fading channels. \color{black} BICM can also use other iterative decoding schemes such as low density parity check code (LDPC). In \cite{Exit4}, it is demonstrated that BICM-ID with signal space diversity (BICM-ID-SSD) outperforms the LDPC-BICM over fading channels. BICM-ID without SSD can also outperform LDPC-BICM when the number of iterations at the decoder is smaller than a certain number, which makes the decoder simpler. Moreover, compared to LDPC-BICM, BICM-ID uses a simple convolutional code instead of a more complex LDPC code. As such, BICM-ID offers a lower system complexity. Consequently, BICM-ID is a good candidate as a coded modulation especially when the system complexity becomes a more important concern. \color{black}
  
 It is widely known that the signal labeling map (mapping) plays a crucial role in BICM-ID performance\cite{8PSK_signaling}. Signal labeling  is defined as the assignment of a binary sequence to a single symbol from a signal constellation. It is also  referred to as multi-dimensional (MD) mapping when a sequence of binary bits is mapped to a vector of symbols instead of a single symbol. 
MD mapping is more flexible to design and also offers   better bandwidth efficiency \cite{MD-BQPSK-Simoen}. The MD labeling process is more adaptable to different design guidelines because MD space provides more diverse Euclidean distances among symbols. In \cite{MD_mapping_TCM}, MD labeling was used for TCM and made the system's BER better through using the available bandwidth more efficiently. This development motivated researchers to use the MD labeling technique and its efficient use of bandwidth to further improve the BER of BICM-ID \cite{MD-BQPSK-Simoen}, \cite{MD_BICMID}-\cite{MD_16_64QAM}.  The MD labeling technique can also be applied to higher order modulation  to increase the data rate of BICM-ID because using a larger constellation makes it possible to send more bits in the same signaling rate. However, providing a suitable MD labeling of a large constellation for BICM-ID is very challenging because of the large number of possible mappings. In general, for a $2N$-D $2^m$-ary modulation, there are $2^{mN}!$ possible mappings, where $!$ denotes the factorial operation. For example, the number of possible $4$-D mappings for a $256$-ary quadrature amplitude modulation ($256$-QAM) is $5.16 \times 10^{287193}$, which is an astronomical figure. \color{black} It is imprtant to note that MD mapping improves the perforamnce of BICM-ID at a particular expense of the system's complexity \cite{MD-Hyper-Ha}-\cite{MD-8PSK-Ha}. However, the study of its complexity is beyond the scope of this paper. 

\color{black}

Different labeling approaches for BICM-ID such as the genetic algorithm (GA), reactive tabu search (RTS) algorithm, extrinsic information transfer (EXIT)-based search algorithm, binary switching algorithm (BSA), and random labeling technique have been extensively investigated \cite{8PSK_signaling}-\cite{High_Order_Map}. Indeed, high computational complexity is the main pitfall of all the proposed computer search based methods  in the literature when looking for a good mapping of a large constellation.

The GA and RTS algorithms have been used in \cite{GA} and \cite{RTS}, respectively, to find the optimum mappings for BICM-ID. In these studies, however,  the authors have not reported results for constellations larger than $64$-QAM due to a very high computational complexity.  In \cite{Exit1}-\cite{Exit3}, an EXIT-based method is proposed to find suitable mappings that improve the iterative decoding systems BER at any signal to noise ratio (SNR) with an arbitrary number of iterations. \color{black} In \cite{Exit4}, the authors have designed an EXIT-chart aided serach method to develop  capacity approaching coded modulations. In particular, they have proposed a BICM-ID system with signal space diversity that approaches the channel capacity in both the fading and non-fading channels. \color{black} However, the EXIT-based search considers both the modulator and encoder in detail and also the iterations between the decoder and demodulator. This makes the process complicated for finding a good mapping of a large MD constellation among a large number of possible labelings.  The binary switching algorithm (BSA) \cite{BSA}, which is the best known mapping search method for BICM-ID, becomes intractable when searching for a good mapping of an MD constellation due to the huge search complexity \cite{MD_BICMID},\cite{MD-8PSK-Ha}. Although it is demonstrated in \cite{MD_BICMID} and \cite{rndm_map} that the random labeling  technique results in suitable mappings for BICM-ID, it is still computationally complex to look for a good mapping of a large constellation. This is because the random labeling method searches for a good mapping among a large set of randomly generated mappings.
 In addition to these computer search based methods, a heuristic method has also been explored in \cite{MD_16_64QAM} to construct MD labeling for BICM-ID. But, this method is limited to $16$- and $64$-QAM. Moreover, this method  is  not
designed for Rayleigh fading channels. 

In this paper, we propose a novel method to develop suitable MD mappings to improve the BER performance of BICM-ID over Rayleigh fading channels in both the low and high SNR regions. To improve the BER performance of the system, we increase the harmonic mean of the minimum squared Euclidean distance (MSED)  \cite{8PSK_signaling},\cite{Chindapol_16QAM} offered by the mapping. 


The reported analytical and numerical results confirm the efficiency of the achieved mappings.
The novelty and contributions of our work are as follows. (i) We break the MD labeling problem into four distinct $2$-D mappings, which makes it easier to optimize the MD labeling. (ii) We design the four $2$-D mappings such that  the resulting MD mapping improves the BER of BICM-ID at low SNRs.  (iii) We develop cost functions that can be optimized over the $2$-D mappings to improve the performance of the resulting MD mapping's at high SNRs.  (iv) We develop efficient MD mappings of higher order modulations including $2^m$-ary ($m= 4, 5, ..., 10$) modulations. Finally, (v) we propose efficient MD mappings of different modulation types  including QAM, phase shift keying (PSK), and irregular modulations. 
\color{black} Compared to the labeling method in \cite{MD_16_64QAM}, 
 the proposed method in this paper has the following additional contributions. (i) The presemnt study is not limited to a particular modulation type, i.e., it is not modulation spesific. However, the method in \cite{MD_16_64QAM} is only for square QAMs. (ii) The proposed method in this paper can generate efficient mappings for modulations with different orders while the method in \cite{MD_16_64QAM} is only for $16$- and $64$-ary particular modulations. (iii) Finally, the proposed method in this paper is designed specifically for Rayleigh fading channels. As a result, over fading channels, the resulting mappings outperform the mappings in \cite{MD_16_64QAM}. \color{black}


The rest of this paper is organized as follows. The BICM-ID system model is described in Section \ref{sysmod}. The proposed MD labeling method is introduced in Section \ref{Porp.Meth}. In Section \ref{cost_func}, cost functions are developed and optimized to obtain good MD mappings. Numerical results and discussions are presented in Section \ref{num_result}.  Finally, Section \ref{conc.} summarizes the conclusions.

\section{System Model}
\label{sysmod}

\begin{figure}[t]
\center
\includegraphics[width= 0.6\columnwidth, viewport = 55mm 207.07mm 135mm 250.00mm]{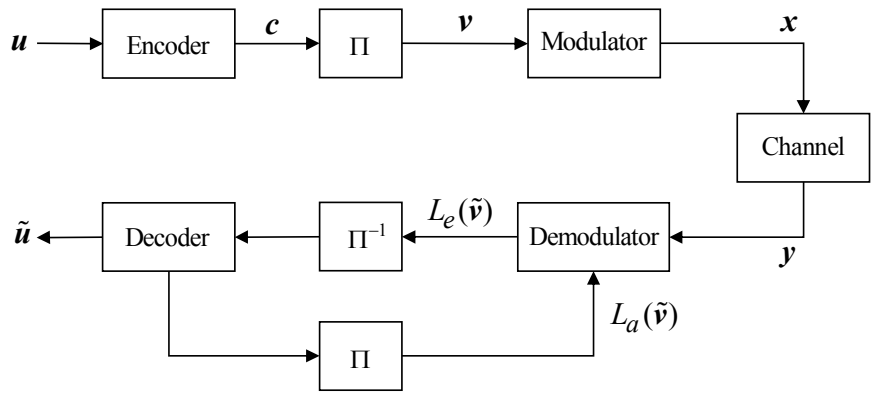}
\caption{The block diagram of a BICM-ID system.}
\label{sys_mod}
\end{figure}

The block diagram of a conventional BICM-ID system is shown in Fig. \ref{sys_mod}. The transmitter is constructed from serial concatenation of a convolutional encoder, a bit interleaver and a modulator. In the system considered in this paper, the modulator maps a sequence of $mN$ bits to a vector of $N$ consecutive $2^{m}$-ary signal points, using a MD mapping function $\mu : \lbrace0, 1\rbrace^{mN}\longrightarrow \mb{\chi} = \chi^{N}$. Let ${\mb x} \in \mb{\chi}$ be a $2N$-D signal point represented as

\begin{equation}
{\mb x} = (x_{1}, x_{2}, \cdots, x_{N}),
\end{equation} where $x_{i} \in \chi$. ${\mb x}$ is labeled by an $mN$-bit binary sequence $\mb{v}$ as

\begin{equation}
{\mb x} = \mu(\mb{v}).
\end{equation}

The average energy per signal-vector is assumed to be $1$, i.e., $E_{\mb{x}} = 1$. It is assumed that the  channel state information is known at the receiver. The received signal-vector corresponding to the transmitted signal-vector $\mb{x}$ is given by

\begin{equation}
\mb{y} = h \mb{x} + \mb{n}, 
\end{equation}where $ h$ is the Rayleigh fading coefficient corresponding to $\mb{x}$, and $\mb{n}$ is a vector of $N$ additive complex white Gaussian noise samples with zero-mean and variance $N_0$. We 
 assume that the fading coefficient for all signals in a signal-vector remains the same. Clearly, for the AWGN channels $h=1$.

At the receiver, from the received signal $y_t$ and the \textit{a priori} log-likelihood ratio (LLR) of the coded bits, the demapper computes the \textit{extrinsic} LLR for each of the  bits in the received symbol as described in \cite{Brink}. Then, the \textit{extrinsic} LLRs are permuted by the random deinterleaver
 and used by the channel
decoder. The decoder then calculates the \textit{extrinsic} LLR on the coded bits using the  BCJR algorithm \cite{BCJR}. These LLRs are interleaved and then fed back to the demapper to be   used as the \textit{a priori} LLRs in the next iteration. 

\section{Proposed mapping method}
\label{Porp.Meth}
As mentioned earlier, for a $2N$-D $2^{m}$-ary modulation there are $2^{mN}!$ possible mappings. In fact, a comprehensive computer search to find good mappings becomes intractable quickly as the  modulation order increases. \color{black} Even \color{black} the well-known BSA mapping search method cannot be used directly to obtain good \color{black} MD mappings of higher order modulations. \color{black} Therefore, we propose an efficient technique to find good MD mappings for  BICM-ID systems over Rayleigh fading channels. 

\subsection{Mapping Design Guideline}
\label{harmonic}
\color{black} The performance of BICM-ID over Rayleigh fading channels is influenced by the \textit{harmonic mean} of the MSED, which is calculated for a given mapping function, $\mu$, applied to signal set $\mb{\chi}$. For a $2N$-D mapping of a $2^{m}$-ary constellation, the \textit{harmonic mean} of the MSED is given by \cite{MD-Hyper-Ha},\cite{Chindapol_16QAM} \color{black}

\begin{equation}
\Phi(\mu, \mb{\chi}) = \left(\dfrac{1}{mN2^{mN}}\sum_{i=1}^{mN} \sum_{b=0}^{1} \sum_{{\mb x}\in \mb{\chi}_{b}^{i}} \dfrac{1}{\Vert {\mb x}- \hat{\mb x} \Vert^{2}}\right) ^{-1},
\label{Harmonic_fading}
\end{equation} where $\mb{x} = (x_{1}, x_{2}, \cdots, x_{N})$ is a $2N$-D signal point and $\mb{\chi}_{b}^{i}$ is the subset of signal points in $\mb{\chi}$ whose labels take value $b$ at the $i^{th}$ bit position. \color{black} For the performance in the low SNR region,  $\hat{\mb x} = (\hat{x}_{1}, \hat{x}_{2}, \cdots, \hat{x}_{N})$ refers to the
nearest neighbor \footnote{Th symbol-vectors with the minimum Euclidean distance are refereed to as the nearest neighbours.} of $\mb{x}$ in $\mb{\chi}_{\bar{b}}^{i}$,
and (\ref{Harmonic_fading}) is referred to as the harmonic
mean of the MSED before feedback. For the asymptotic
performance (performance in the high SNR region), $\mb{\chi}_{\bar{b}}^{i}$
involves only one symbol-vector $\hat{\mb x}$, which
is different from $\mb{x}$ only in the $i^{th}$ bit position \cite{Chindapol_16QAM}. In this
case, (\ref{Harmonic_fading}) is referred to as the harmonic mean of the MSED after
feedback, which is denoted by $\hat\Phi(\mu, \mb{\chi})$. To achieve good
performance in the low and high SNR regions, a large value of $\Phi(\mu, \mb{\chi})$ and $\hat\Phi(\mu, \mb{\chi})$ is required, respectively. However, maximizing (\ref{Harmonic_fading})  is a very  complex \color{black} problem even for  a modulation such as $64$-QAM \cite{RTS}. We propose an innovative approach to generate MD mappings using $2$-D mappings such that the resulting MD mapping has a greater value of $\Phi(\mu, \mb{\chi})$. \color{black} Later, we develop cost functions that are optimized over the employed $2$-D mappings to achieve a high value of $\hat{\Phi}(\mu, \mb{\chi})$ for the MD mapping. Our cost functions are very simple and give excellent results, even for higher order constellations such as MD $1024$-QAM.



\subsection{MD Mapping Using $2$-D Mappings}
\label{three_b}
 Let  $\mb{l} = (l_{1}, l_{2}, \cdots, l_{mN})$ be an $mN$-bit binary label. We can write $\mb{l} = (\mb{l}_{1}, \mb{l}_{2}, \cdots, \mb{l}_{N})$, where $\mb{l}_{i}$ is a $m$-bit binary label and is given by

\begin{equation}
\mb{l}_{i} = (l_{(i-1)m+1}, \cdots, l_{im});~~ i = 1, ..., N.
\end{equation} 
Suppose that $\mathcal{L}$ denotes the set of all $mN$-bit binary labels and $\mathcal{L}_{e}$ and $\mathcal{L}_{o}$  represent the subset of all $\mb{l}\in{\mathcal{L}}$ with even and odd Hamming weights, respectively. The MD mapping problem can be broken into four mappings in
2-D signal space as described below. According to the proposed MD mapping function, i.e., $\mu$, label $\mb{l}$ is mapped to the $2N$-D signal point $\mb{x} = (x_{1}, \cdots, x_{N})$ as given below

\begin{eqnarray}
\label{mapping_func}
x_{i} =\left\{
\begin{array}{llll}
\lambda_{el}(\mb{l}_{i}) & \mbox{if}~ i=1 ~ \mbox{, }~ \mb{l} \in \mathcal{L}_{e}, \\
\lambda_{ol}(\mb{l}_{i}) & \mbox{if}~ i=1 ~ \mbox{, }~ \mb{l} \in \mathcal{L}_{o},\\
\lambda_{er}(\mb{l}_{i}) & \mbox{if}~ i \geqslant 2 ~ \mbox{, }~ \mb{l} \in \mathcal{L}_{e},\\
\lambda_{or}(\mb{l}_{i}) & \mbox{if}~ i \geqslant 2 ~ \mbox{,}~ \mb{l} \in \mathcal{L}_{o},\\
\end{array}; ~~ i = 1, ..., N,\right.
\end{eqnarray}
where $\lambda_{el}$, $\lambda_{ol}$, $\lambda_{er}$, and $\lambda_{or}$ are $2$-D mapping functions, which will be discussed later in this section.
In the applied mapping, let $\mb{\chi}_{e}$ and $\mb{\chi}_{o}$ represent the subset of signal points in $\mb{\chi}$ whose labels belong to $\mathcal{L}_{e}$ and $\mathcal{L}_{o}$, respectively. Without loss of generality, assume that $\mb{x} \in \mb{\chi}_{e}$ and $\hat{\mb{x}} \in \mb{\chi}_{o}$ where $\hat{\mb x} = (\hat{x}_{1}, \hat{x}_{2}, \cdots, \hat{x}_{N})$ is a signal point whose label is different from that of $\mb{x}$ only in one bit position.  
We partition the $2$-D signal constellation $\chi$ into two separate subsets with equal cardinalities and denote them as $\chi_{el}$ and $\chi_{ol}$. Then, we limit the first element in $\mb{x}$ and $\hat{\mb{x}}$, i.e., $x_{1}$ and $\hat{x}_{1}$, to belong to $\chi_{el}$ and $\chi_{ol}$, respectively. 
In (\ref{mapping_func}), $\lambda_{el}(.)$ and $\lambda_{ol}(.)$ each map an $m$-bit label to a $2$-D signal point chosen from $\chi_{el}$ and $\chi_{ol}$, respectively. However, $\chi_{el}$ and $\chi_{ol}$ involve only $2^{m-1}$ signal points while there are $2^{m}$ distinct $m$-bit labels. As a consequence, each signal point in $\chi_{el}$ and $\chi_{ol}$ should be mapped by two $m$-bit labels simultaneously. In order to obtain a one-to-one MD mapping function, we restrict the two labels that are mapped to a particular signal point in either $\chi_{el}$ or $\chi_{ol}$ to be different in an odd number \color{black} of bits. Specifically, we assume they are different only in the first bit position. This prohibits the labels with a Hamming distance larger than ($m+1$) bits to be mapped to the nearest neighbours in either of $\mb{\chi}_{e}$ or $\mb{\chi}_{o}$. As a result, the value of $\Phi(\mu, \mb{\chi})$ increases.  \color{black} 
On the other hand, there is no constraint on $\lambda_{er}(.)$ and $\lambda_{or}(.)$ except they need to be bijective.

\begin{prop}
\label{prop1}
Let us assume that $\mb{l}_1 = (l_1^1, l_1^2, \cdots, l_1^n)$ and $\mb{l}_2 = (l_2^1, l_2^2, \cdots, l_2^n)$ are two $n$-bit labels,  $D = d_H(\mb{l}_1,\mb{l}_2)$ is the Hamming distance between $\mb{l}_1$ and $\mb{l}_2$, and $W = w_1 + w_2$, where $w_j$ is the Hamming weight of $\mb{l}_j$ ($j = 1, 2$). If $W \in \mathbb{E}$, $D \in \mathbb{E}$, and if  $W \in \mathbb{E}$, $D \in \mathbb{E}$. 
 \end{prop}

\begin{proof}
See Appendix \ref{apdx1}.

\end{proof}

\begin{prop}
\label{prop2}
In the proposed MD mapping function, $\mu$, there is a one-to-one correspondence between MD signal points and binary labels. \end{prop}

\begin{proof}
See Appendix \ref{apdx2}.

\end{proof}

\color{black}
\begin{prop}
\label{prop3}
In the proposed MD mapping function, the Hamming distance between the nearest neighbours in either of $\mb{\chi}_{e}$ or $\mb{\chi}_{o}$ cannot be larger than ($m+1$) bits. 
\end{prop}

\begin{proof}
See Appendix \ref{apdx3}.

\end{proof}
\color{black}


\section{Development and optimization of cost functions to improve $\hat{\Phi}$}
\label{cost_func}
 As mentioned in Section \ref{harmonic}, maximizing the harmonic mean is a very complicated  task. In this section, we use the  $2$-D mappings in (\ref{mapping_func}), i.e., $\lambda_{el}$, $\lambda_{ol}$, $\lambda_{er}$, $\lambda_{or}$, to develop two cost functions and derive a lower bound for $\hat{\Phi}(\mu, \mb{\chi})$, i.e., $\Delta$. Then, we propose an algorithm to maximize $\Delta$ by optimizing the cost functions over the $2$-D mappings. As such, the achieved $2$-D mappings construct a MD mapping that provides a larger value of $\Delta$, which results in a larger value of  $\hat{\Phi}(\mu, \mb{\chi})$. 


\subsection{Development of Cost Functions}
\color{black} 
 In order to achieve a lower bound for $\hat{\Phi}(\mu, \mb{\chi})$, $\Delta$, we develop an upper bound for $\hat{\Phi}(\mu, \mb{\chi})^{-1}$. In particular, we first decompose $\hat{\Phi}^{-1}$ into two equal parts, i.e., $\Omega_e (\mu, \mb{\chi})$ and $\Omega_o(\mu, \mb{\chi})$, where in $\Omega_e$, $\mb{x}\in \mb{\chi}_e$ and $\hat{\mb{x}} \in \mb{\chi}_{o}$. Next, we derive an upper bound for $\Omega_e (\mu, \mb{\chi})$, i.e., $\Psi(\mu, \mb{\chi})$, and then, we decompose $\Psi$ into $\Psi_l(\mu, \mb{\chi})$ and $\Psi_r(\mu, \mb{\chi})$, where $\Psi_l$ uses only the first symbol in $\mb{x}$ and $\hat{\mb{x}}$, i.e., $x_1$ and  $\hat{x}_1$, and $\Psi_l$ uses the rest of symbols in $\mb{x}$ and $\hat{\mb{x}}$. As $\mb{x}\in \mb{\chi}_e$ and $\hat{\mb{x}} \in \mb{\chi}_{o}$, $x_1$ and $\hat{x}_1$ in $\Psi_l$ are obtained using the mapping functions $\lambda_{el}$ and $\lambda_{ol}$, respectively, while $x_i$ and $\hat{x}_i$ ($2 \leqslant i \leqslant N$) in $\Psi_r$ are obtained using $\lambda_{er}$ and $\lambda_{or}$, respectively. Thus, we develop two cost functions $\psi_{l}$ and $\psi_{r}$, where $\psi_{l}$ generates the same values as $\Psi_l$ by considering all the cases of using $\lambda_{el}$ and $\lambda_{ol}$ in $\Psi_l$, and $\psi_{r}$ generates the same vales as $\Psi_r$ by considering all the cases of using $\lambda_{er}$ and $\lambda_{or}$ in $\Psi_r$. Finally, we use $\psi_{l}$ and $\psi_{r}$ to develop $\Delta$. In what follows, we discuss these steps in more detail.

\color{black}
We use (\ref{Harmonic_fading}) to write

\begin{dmath}
\label{Harmonic_fading2}
\hat{\Phi}(\mu, \mb{\chi})^{-1} = \Omega_{e} (\mu, \mb{\chi}) + \Omega_{o}(\mu, \mb{\chi}),
\end{dmath}  
where

\begin{dmath}
\label{omega_e}
 \Omega_{e} (\mu, \mb{\chi}) =  \dfrac{1}{mN2^{mN}} 
\sum_{i=1}^{mN} \sum_{b=0}^{1} \sum_{\substack{\mb{x}\in \mb{\chi}_{b}^{i} \\ \mb{x}\in \mb{\chi}_{e}}} \dfrac{1}{\Vert {\mb x}- \hat{\mb x} \Vert^{2}}  
\end{dmath}
and

\begin{dmath}
\label{omega_o}
 \Omega_{o} (\mu, \mb{\chi}) = \dfrac{1}{mN2^{mN}} 
\sum_{i=1}^{mN} \sum_{b=0}^{1} \sum_{\substack{\mb{x}\in \mb{\chi}_{b}^{i} \\ \mb{x}\in \mb{\chi}_{o}}} \dfrac{1}{\Vert {\mb x}- \hat{\mb x} \Vert^{2}}.
\end{dmath} 
The sets $\mb{\chi}_{e}$ and $\mb{\chi}_{o}$ have the same cardinality. Moreover, when a given $\mb{x}$ is in $\mb{\chi}_{e}$ then the corresponding $\hat{\mb{x}}$ belongs to $\mb{\chi}_{o}$ and vice versa. Therefore, (\ref{omega_e}) and (\ref{omega_o}) use the same set of Euclidean distances, which results in $\Omega_{e}(\mu, \mb{\chi}) = \Omega_{o}(\mu, \mb{\chi})$. As a result, using (\ref{Harmonic_fading2}) we have 
\color{black} 
\begin{equation}
\label{Harmonic_omega}
\hat{\Phi}(\mu, \mb{\chi})^{-1} = 2\Omega_{e} (\mu, \mb{\chi}).
\end{equation} 
 Since $\Vert \mb{x}- \hat{\mb{x}} \Vert^{2} = \sum _{j = 1}^{N} \vert x_{j}- \hat{x}_{j} \vert^{2}$, then (\ref{omega_e}) can be rewritten as

\begin{dmath}
\label{Harmonic_omega2}
\Omega_{e} (\mu, \mb{\chi}) = \dfrac{1}{mN2^{mN}} \sum_{i=1}^{mN} \sum_{b=0}^{1} \sum_{\substack{\mb{x}\in \mb{\chi}_{b}^{i} \\ \mb{x}\in \mb{\chi}_{e}}} \dfrac{1}{\sum_{j = 1}^{N}\vert x_{j} - \hat{x}_{j} \vert^{2}}.\\ \nonumber
\end{dmath}

\begin{prop}
\label{prop4}
Let $\mb{y} = (y_{1}, y_{2}, \cdots, y_{N})$ be a vector of non-negative real numbers. Then, we have 

\begin{equation}
\frac{1}{\sum_{i = 1}^{N} y_{i} } \leqslant \frac{1}{N} \sum_{j = 1}^{N} \frac{1}{y_{j}}.
\label{Prop1_Inq1}
\end{equation}
\end{prop}

\begin{proof}
See Appendix \ref{apdx4}.

\end{proof}

Applying (\ref{Prop1_Inq1}) in (\ref{Harmonic_omega2}), we can write

\begin{equation}
\label{Harmonic_omega3}
\Omega_{e} (\mu, \mb{\chi}) \leqslant K\Psi(\mu, \mb{\chi})
\end{equation} where $K=\frac{1}{mN^22^{mN}}$ is a constant value and $\Psi(\mu, \mb{\chi})$ is defined as

\begin{equation}
\label{Harmonic_psi1}
\Psi(\mu, \mb{\chi}) = \sum_{i=1}^{mN} \sum_{b=0}^{1} \sum_{\substack{\mb{x}\in \mb{\chi}_{b}^{i} \\ \mb{x}\in \mb{\chi}_{e}}} \sum_{j = 1}^{N} \dfrac{1}{\vert x_{j} - \hat{x}_{j} \vert^{2}}.\\ \nonumber
\end{equation} We can decompose $\Psi(\mu, \mb{\chi})$ as

\begin{equation}
\label{Harmonic_psi2}
\Psi(\mu, \mb{\chi}) = \Psi_{l}(\mu, \mb{\chi}) + \Psi_{r}(\mu, \mb{\chi}),
\end{equation} where $\Psi_{l}(\mu, \mb{\chi})$ and $\Psi_{r}(\mu, \mb{\chi})$ are given by 

\begin{equation}
\label{Harmonic_psi_l}
\Psi_{l}(\mu, \mb{\chi}) = \sum_{i=1}^{mN} \sum_{b=0}^{1} \sum_{\substack{\mb{x}\in \mb{\chi}_{b}^{i} \\ \mb{x}\in \mb{\chi}_{e}}} \dfrac{1}{\vert x_{1} - \hat{x}_{1} \vert^{2}}
\end{equation}

and

\begin{equation}
\label{Harmonic_psi_r}
\Psi_{r}(\mu, \mb{\chi}) = \sum_{i=1}^{mN} \sum_{b=0}^{1} \sum_{\substack{\mb{x}\in \mb{\chi}_{b}^{i} \\ \mb{x}\in \mb{\chi}_{e}}} \sum_{j = 2}^{N} \dfrac{1}{\vert x_{j} - \hat{x}_{j} \vert^{2}}.
\end{equation} 

Let $\mb{l} = (l_{1}, l_{2}, \cdots, l_{mN})$ and $\hat{\mb{l}} = (\hat{l}_{1}, \hat{l}_{2}, \cdots, \hat{l}_{mN})$ are two $mN$-bit labels, which are different only in the $i^{th}$ bit position, and are mapped to  $\mb{x} = (x_{1}, \cdots, x_{N})$ and $\hat{\mb{x}} = (\hat{x}_{1}, \cdots, \hat{x}_{N})$, respectively. We define $\mb{l}_{i}$ and $\tilde{\mb{l}}_{i}$ respectively as the $i^{th}$ $m$-tuple bits of $\mb{l}$ and $\hat{\mb{l}}$, and rewrite $\mb{l} = (\mb{l}_{1}, \mb{l}_{2}, \cdots, \mb{l}_{mN})$ and $\tilde{\mb{l}} = (\tilde{\mb{l}}_{1}, \tilde{\mb{l}}_{2}, \cdots, \tilde{\mb{l}}_{mN})$. 
Then, $\Psi_{l}$ in (\ref{Harmonic_psi_l}) is equal to $\Psi'_{l}$, which is given by

\begin{equation}
\label{Harmonic_psi_l2}
\Psi'_{l}(\lambda_{el}, \lambda_{ol}, \mathcal{L}) = \sum_{i=1}^{mN} \sum_{b=0}^{1} \sum_{\substack{\mb{l}\in \mathcal{L}_{b}^{i} \\ \mb{l}\in \mathcal{L}_{e}}} \dfrac{1}{\vert \lambda_{el}(\mb{l}_{1}) - \lambda_{ol}(\hat{\mb{l}}_{1}) \vert^{2}},
\end{equation}
where $\mathcal{L}_{b}^{i} \in \mathcal{L}$ is the subset of labels with value $b$ in their $i^{th}$ bit position.  For a given $m$-bit sequence  $\mb{l}_{i}$, $\hat{\mb{l}}_{i}$ can take $(m+1)$ distinct $m$-bit sequences, where each one is the same as $\mb{l}_{i}$ or different from $\mb{l}_{i}$ only in one bit position. For example, if $m = 4$ and $\mb{l}_{i} = (0,0,0,0)$, $\hat{\mb{l}}$ can take either of the $5$ labels in $\lbrace (0,0,0,0),(0,0,0,1),(0,0,1,0),(0,1,0,0),(1,0,0,0)\rbrace$. Let $\mb{\alpha} = (\alpha_{1}, \cdots, \alpha_{m})$ and $\mb{\beta} = (\beta_{1}, \cdots, \beta_{m})$ be two binary sequences, where $\mb{\beta}$ has the Hamming distance of either zero or one bit from $\mb{\alpha}$. The set of $(m+1)$ possibilities for $\mb{\beta}$ is denoted by $\mathcal{B}$. Assume that for a given $i$, $\mb{l}_{i}=\mb{\alpha}$ and $\hat{\mb{l}}_{i}=\mb{\beta}$.  Then, $\Psi'_{l}$  in (\ref{Harmonic_psi_l2}) is equal to $\psi_{l}$, which is defined as 

\begin{equation}
\label{Harmonic_psi_l3}
\psi_{l}(\lambda_{el}, \lambda_{ol},\chi_{el}, \chi) = \sum_{\mb{\alpha}} \sum_{\mb{\beta} \in \mathcal{B}} \dfrac{a_{\mb{\alpha},\mb{\beta}}^{(l)}}{\vert \lambda_{el}(\mb{\alpha}) - \lambda_{ol}(\mb{\beta}) \vert^{2}},
\end{equation} where $a_{\mb{\alpha},\mb{\beta}}^{(l)}$ is computed as 

\begin{equation}
\label{a_l}
a_{\mb{\alpha},\mb{\beta}}^{(l)} = \sum_{i=1}^{mN} \sum_{b=0}^{1} \sum_{\substack{\mb{l}\in \mathcal{L}_{b}^{i} \\ \mb{l}\in \mathcal{L}_{e}}} I(\mb{l}_{1} = \mb{\alpha}, \hat{\mb{l}}_{1} = \mb{\beta}),
\end{equation} where $I(x)$ is an indicator function and is defined as 
\begin{equation}
I(x) =\left\{
\begin{array}{ll}
1, & \mbox{if}~ x ~ \mbox{is true},\\
0, & \mbox{otherwise}.
\end{array}\right.
\end{equation} 
Similarly, $\Psi_{r}$ in (\ref{Harmonic_psi_r}) is equal to $\Psi'_{r}$, which is given by 

\begin{equation}
\label{Harmonic_psi_r2}
\Psi'_{r}(\lambda_{er}, \lambda_{or}, \mathcal{L}) = \sum_{i=1}^{mN} \sum_{b=0}^{1} \sum_{\substack{\mb{l}\in \mathcal{L}_{b}^{i} \\ \mb{l}\in \mathcal{L}_{e}}} \sum_{j = 2}^{N} \dfrac{1}{\vert \lambda_{er}(\mb{l}_{j}) - \lambda_{or}(\hat{\mb{l}}_{j}) \vert^{2}}.
\end{equation} The $m$-bit elements $\mb{l}_{i}$ in $\mb{l} = (\mb{l}_{1}, \mb{l}_{2}, \cdots, \mb{l}_{N})$ are independent from one another for all values of $i$. Then, $\Psi'_{r}$ in (\ref{Harmonic_psi_r2}) is equal to   $\psi_{r}$, which is defined as
\begin{equation}
\label{Harmonic_psi_r3}
\psi_{r}(\lambda_{er}, \lambda_{or}, \chi) = \sum_{i=1}^{mN} \sum_{b=0}^{1} \sum_{\substack{\mb{l}\in \mathcal{L}_{b}^{i} \\ \mb{l}\in \mathcal{L}_{e}}} \dfrac{N-1}{\vert \lambda_{er}(\mb{l}_{2}) - \lambda_{or}(\hat{\mb{l}}_{2}) \vert^{2}},
\end{equation} 
and can be rewritten as

\begin{equation}
\label{Harmonic_psi_r3}
\psi_{r}(\lambda_{er}, \lambda_{or}, \chi) = \sum_{\mb{\alpha}} \sum_{\mb{\beta} \in \mathcal{B}} \dfrac{(N-1)a_{\mb{\alpha},\mb{\beta}}^{(r)} }{\vert \lambda_{er}(\mb{\alpha}) - \lambda_{or}(\mb{\beta}) \vert^{2}},
\end{equation} where $a_{\mb{\alpha},\mb{\beta}}^{(r)}$ is computed as 

\begin{equation}
\label{a_l}
a_{\mb{\alpha},\mb{\beta}}^{(r)} = \sum_{i=1}^{mN} \sum_{b=0}^{1} \sum_{\substack{\mb{l}\in \mathcal{L}_{b}^{i} \\ \mb{l}\in \mathcal{L}_{e}}} I(\mb{l}_{2} = \mb{\alpha}, \hat{\mb{l}}_{2} = \mb{\beta}).
\end{equation}

Using (\ref{Harmonic_omega}), (\ref{Harmonic_omega3}), and (\ref{Harmonic_psi2}), a lower bound of  $\hat{\Phi}(\mu, \mb{\chi})$ can be derived as follows

\begin{eqnarray}
\label{lowr_bnd1}
& \hat{\Phi}^{-1} (\mu, \mb{\chi}) \leqslant 2K \left(\Psi_{l}(\mu, \mb{\chi}) + \Psi_{r}(\mu, \mb{\chi})\right) \\ \nonumber
& \Rightarrow \Delta \leqslant \hat{\Phi} (\mu, \mb{\chi}),
\end{eqnarray}
where $\Delta$ is given by
\begin{eqnarray}
\label{lowr_bnd12}
\Delta = \dfrac{1}{2K(\Psi_{l}(\mu, \mb{\chi}) + \Psi_{r}(\mu, \mb{\chi}))}.
\end{eqnarray} Because $\Psi_{l}(\mu, \mb{\chi}) = \psi_{l}(\lambda_{el}, \lambda_{ol}, \chi_{el}, \chi)$  and $\Psi_{r}(\mu, \mb{\chi}) = \psi_{r}(\lambda_{er}, \lambda_{or}, \chi)$, we rewrite $\Delta$ as

\begin{eqnarray}
\label{lowr_bnd12}
\Delta = \dfrac{1}{2K(\psi_{l}(\lambda_{el}, \lambda_{ol}, \chi_{el}, \chi) + \psi_{r}(\lambda_{er}, \lambda_{or}, \chi))}.
\end{eqnarray} Note that (\ref{lowr_bnd12}) operates in $2$-D signal space rather than MD signal space, and as a result, optimization is  much simpler.

Our objective is to maximize $\Delta$ and then to calculate the corresponding $\hat{\Phi} (\mu, \mb{\chi})$. Since $\psi_{l}(\lambda_{el}, \lambda_{ol}, \chi_{el}, \chi) $ and $ \psi_{r}(\lambda_{er}, \lambda_{or}, \chi)$ are independent from each other, then the maximum value of $\Delta$, $\Delta_{max}$, is given by

\begin{dmath}
\label{delta_max}
\Delta_{max} = \dfrac{1}{2K}\left[\min_{\lambda_{el}, \lambda_{ol},\chi_{el}} \psi_{l}(\lambda_{el}, \lambda_{ol}, \chi_{el}, \chi) + \min_{\lambda_{er}, \lambda_{or}} \psi_{r}(\lambda_{er}, \lambda_{or}, \chi)\right]^{-1}.
\end{dmath}

\subsection{Minimization of Cost Functions}
\label{four_b}
 The minimization of the cost functions $\psi_{l}$ and $\psi_{r}$ in (\ref{delta_max}) can be done using the BSA \cite{BSA}. Note that the minimization of $\psi_{r}$ is simpler than that of $\psi_{l}$ because $\psi_{r}$ deals with fewer effective arguments. Thus, we first optimize $\psi_{r}$, and then, use the obtained results to simplify the minimization of $\psi_{l}$. 
 \subsubsection{Minimization of $\psi_{r}$}
 To minimize $\psi_{r}$, two random mappings are initially  considered as $\lambda_{er}$ and $\lambda_{or}$. Then, the BSA is used to minimize $\psi_{r}$ by modifying $\lambda_{er}$ and $\lambda_{or}$.  In fact in $\psi_{r}$, the cost value for a given symbol in $\lambda_{er}$ is computed by using $(m+1)$ corresponding symbols from $\lambda_{or}$. As a result,  simultaneously modifying $\lambda_{er}$ and $\lambda_{or}$  can make the optimization complex. Therefore, our approach is to minimize $\psi_{r}$ by alternatingly modifying each of $\lambda_{er}$ and $\lambda_{or}$. 
In other words, we use the BSA to decrease $\psi_{r}$ by modify $\lambda_{er}$. After a given number of iterations, $\lambda_{er}$ and $\lambda_{or}$ are exchanged. Again, the BSA is used to decrease $\psi_{r}$ by modifying the new $\lambda_{er}$. 
This procedure is repeated up to a given number of iterations. 

 \subsubsection{Minimization of $\psi_{l}$}
 In addition to $\lambda_{el}$ and $\lambda_{ol}$, $\chi_{el}$ is another effective argument in computing $\psi_{l}$. As there is no constraint on $\chi_{el}$, it is a complex process to minimize $\psi_{l}$. To simplify the optimization process, $\chi_{el}$ is constrained  to involve only the symbols whose labels in the obtained $\lambda_{er}$ take binary value $b$ in a given bit position. In this paper, we assume that $\chi_{el}$ involves the symbols whose labels in $\lambda_{er}$ take the value zero in the first bit-position. The functions $\psi_{l}$ and $\psi_{r}$ are computed by considering the similar Euclidean distances between two-dimensional symbols. As a result, there is a potential advantage in applying the above mentioned constraint on $\chi_{el}$ because it will be easier to find a suitable $\lambda_{el}$ corresponding to a given $\lambda_{ol}$. 
After determining  $\chi_{el}$ and $\chi_{ol}$, two random mappings are generated as $\lambda_{el}$ and $\lambda_{ol}$. 
Again the BSA is applied to minimize $\psi_{l}$ by modifying $\lambda_{el}$. Then, $\lambda_{el}$ is exchanged by $\lambda_{ol}$ and the BSA minimizes $\psi_{l}$ by modifying the new $\lambda_{el}$. This procedure is repeated up to a given number of iterations. \color{black}

By executing the proposed  algorithm for a certain number of iterations, a local maximum value is calculated using (\ref{delta_max}). The search algorithm is executed several times and each time the corresponding value for $\hat{\Phi}(\mu, \mb{\chi})$ is calculated. Finally, the modulations corresponding to the maximum obtained $\hat{\Phi}(\mu, \mb{\chi})$ are chosen. Fig. \ref{flwchrt} illustrates the flowchart of the proposed algorithm.
Numerical results confirm that the proposed algorithm generates mappings with significantly large values of $\hat{\Phi}(\mu, \mb{\chi})$. As a result, the obtained mappings would also improve the  error performance of BICM-ID systems in the high SNR region over Rayleigh fading channels.

\subsection{Simplified BSA}
\color{black}  We use the simplified BSA in our algorithm to reduce the search complexity. The BSA calculates a cost function for each symbol in an initial random mapping, and then, lists the symbols in descending order in terms of cost value. Next, the label of the symbol with the highest cost value is switched with the label of another symbol such that the total cost is reduced as much as possible. After each switch, the BSA again lists the symbols in descending cost value and repeats the switching process. 
\begin{figure}[t]
\center
\includegraphics[width= 0.51\columnwidth, viewport=50.40mm 50mm 144.25mm 242.00mm]{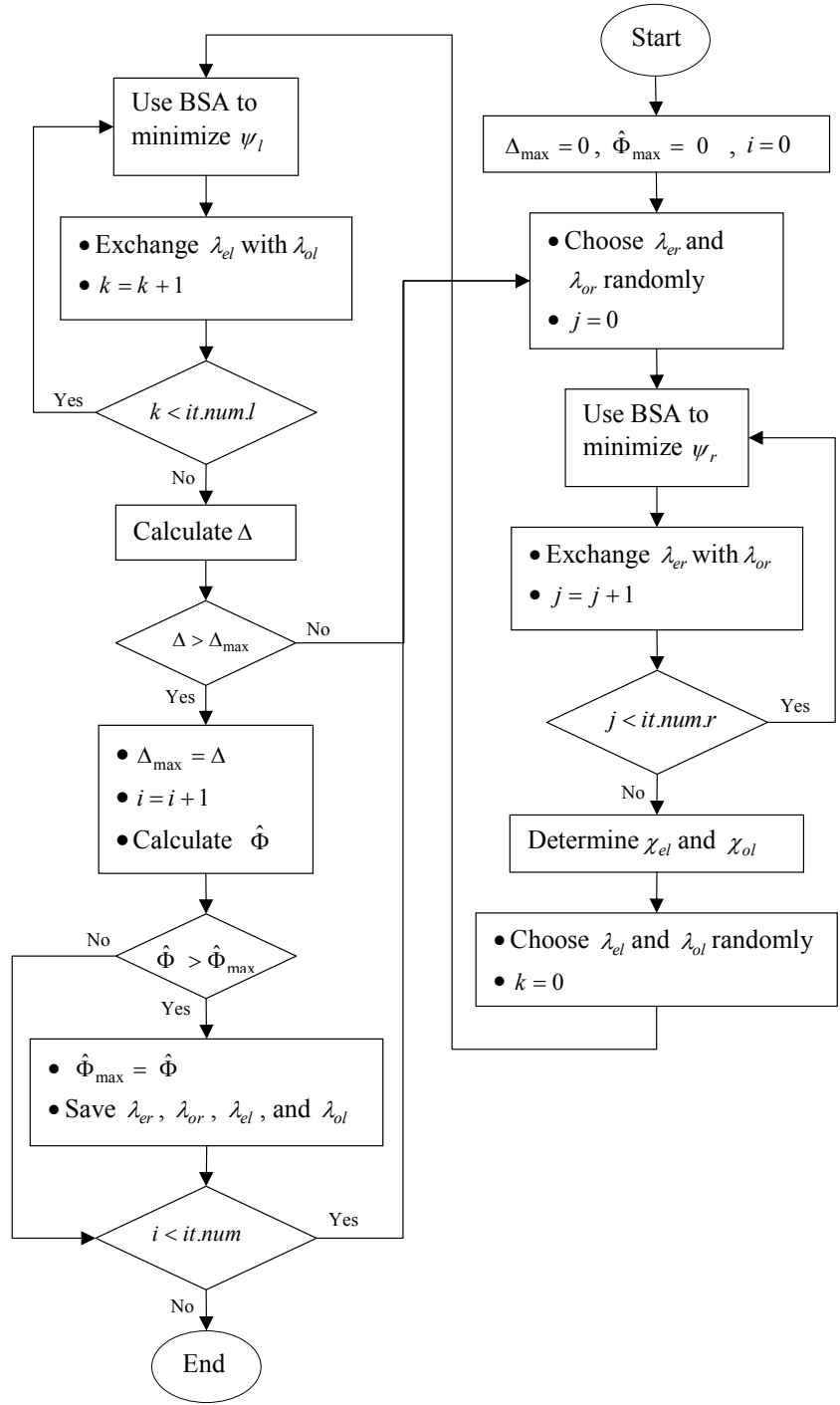}
\caption{Flowchart of the proposed algorithm (\textit{it.num.r}, \textit{it.num.l}, and \textit{it.num} represent the number of iterations for different loops).}
\label{flwchrt}
\end{figure}
It is important to note that each switch in the subsequent rounds of the BSA affects the cost value of only a limited number of symbols in the mapping. In particular, in maximizing $\hat{\Phi}(\mu, \mb{\chi})$ for a given modulation using the BSA, when the label of symbol $s_{i}$, i.e., $\mb{l}_{i}$, is switched by label of symbol $s_{j}$, i.e., $\mb{l}_{j}$, the cost value  changes only for $s_{i}$ and $s_{j}$ and for the symbols whose labels are different from $\mb{l}_{i}$ or $\mb{l}_{j}$ only in one bit position. For example, for a $2^{m}$-ary modulation, the number of affected symbols in each switch in a subsequent round is at most ($m+2$). Therefore, in the subsequent rounds of the simplified BSA, we calculate the cost function only for the symbols that are affected by the most recent switch. This modification makes the BSA much simpler while resulting in the exactly same results. \color{black}


\color{black}

\begin{table*}[b]
\caption{Proposed $\lambda_{er}$, $\lambda_{or}$, $\lambda_{el}$,  and $\lambda_{ol}$ for $8$-ary modulations.}
\centering
\resizebox{0.73\columnwidth}{!}{%
\begin{tabular}{c|cccccccc|cccccccc|}
\cline{2-17}
 &\multicolumn{8}{c|}{$8$-PSK} & \multicolumn{8}{c|}{Optimum $8$-QAM}\\
\hline
\multicolumn{1}{|c|}{$\lambda_{er}$} & [2 & 7 & 6 & 5 & 4 & 1 & 0 & 3] & [1 & 3 & 5 & 7 & 2 & 0 & 6 & 4] \\ \hline
\multicolumn{1}{|c|}{$\lambda_{or}$} & [4 & 1 & 0 & 3 & 2 & 7 & 6 & 5] & [6 & 4 & 2 & 0 & 5 & 7 & 1 & 3]\\ \hline

\multicolumn{1}{|c|}{$\lambda_{el}$} & [(2, & 6) & (1, & 5) & (0, & 4) & (3, & 7)] & [(1, & 5) & (3, & 7) & (2, & 6) & (0, & 4)]\\ \hline
\multicolumn{1}{|c|}{$\lambda_{ol}$} & [(1, & 5) & (0, & 4) & (3, & 7) & (2, & 6)]  & [(2, & 6) & (3, & 7) & (0, & 4) & (1, & 5)]\\\hline
\end{tabular}
\label{MD_8ary_Map}
}
\end{table*}
\section{Numerical Results and Discussion}
\label{num_result}

\begin{figure}[t]
\centering
\includegraphics[width= 0.48\columnwidth, viewport=25.40mm 170.13mm 159.51mm 235.00mm]{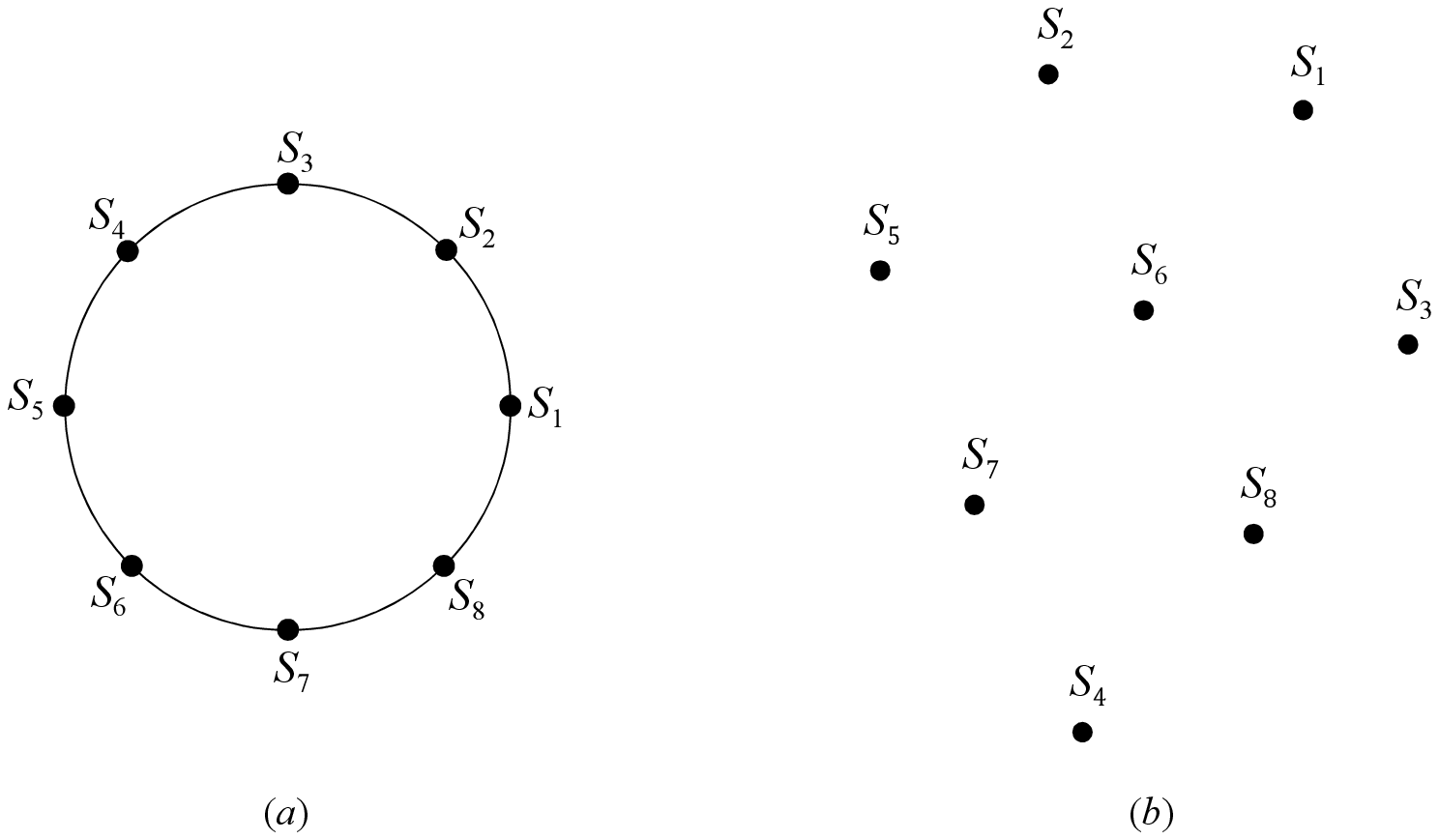}
\caption{(a). Symbol's arrangement in $8$-ary constellations: (a). $8$-PSK and (b). the optimum $8$-QAM \cite{optimal_8QAM}.}
\label{8ary_cons}
\end{figure}

This section provides resulting mappings and selected numerical results to illustrate the performance and advantage of our proposed MD mappings for BICM-ID. 

\subsection{Resulting MD Mappings}
\color{black} Our proposed algorithm  is used to obtain MD mappings of different modulations such as $8$-PSK, the optimum $8$-QAM \cite{optimal_8QAM}, and $2^{m}$-QAM for $m= 4, 5, \cdots, 10$. \color{black} Tables \ref{MD_8ary_Map}-\ref{MD_16QAM_Map} show the resulting $2$-D mappings, i.e., $\lambda_{er}$, $\lambda_{or}$, $\lambda_{el}$, and $\lambda_{ol}$, in decimal format for $8$-ary modulations and $16$-QAM. The results for larger modulations are repoertd in Apendix \ref{app_mappings}. In these tables, the decimal labels are ordered according to the symbol order in the corresponding constellations. For  $8$-PSK and the optimum $8$-QAM, the considered symbol order is shown in Fig. \ref{8ary_cons}. For square QAMs, it is assumed that the symbol order starts from the top left corner in the constellation and increases from top to bottom and from left to right (see Fig. \ref{16QAM}.(a) as an example for $16$-QAM). For cross QAM constellations, such as $32$-QAM, we consider the symbol order used in \cite{Cross_QAM}.  
In these tables, the resulting $2$-D mappings for higher order modulations are indicated in multiple rows. For  example in Table \ref{MD_64QAM_Map}, $\lambda_{er}$ for $64$-QAM is indicated in two rows, where the first element in the second row is the label of the $33^{rd}$ symbol in the $64$-QAM constellation.  In addition, two labels in the $i^{th}$ parentheses in  $\lambda_{el}$ and $\lambda_{ol}$ in each table belong to the $i^{th}$ symbol in $\chi_{el}$ and $\chi_{ol}$, respectively. For example, Fig. \ref{16QAM}.(b), (c), and (d) illustrate the $16$-QAM mappings reported in Table \ref{MD_16QAM_Map}. As mentioned in Section \ref{four_b},  $\chi_{el}$ involves the symbols whose binary labels in $\lambda_{er}$ take the value zero at the first bit position. In other words, $\chi_{el}$ is constructed by the symbols whose decimal label in $\lambda_{er}$ is smaller than $\frac{M}{2}$. As a result, for $16$-QAM, $\chi_{el} = \lbrace S_{1}, S_{2}, S_{5}, S_{6}, S_{9}, S_{10}, S_{13}, S_{14} \rbrace$, where $\chi_{el}$ is indicated by unshaded symbols in Fig. \ref{16QAM}.(d). The remaining $16$-QAM symbols belong to $\chi_{ol}$, which are shaded in Fig. \ref{16QAM}.(d). Example \ref{exmpl1} clarifies how to use $\lambda_{el}$, $\lambda_{ol}$, $\lambda_{er}$, and $\lambda_{or}$ to construct the proposed MD mapping of $16$-QAM.  

\begin{figure}[t]
\centering
\includegraphics[width= 0.8\columnwidth, viewport=20mm 215mm 165mm 254.00mm]{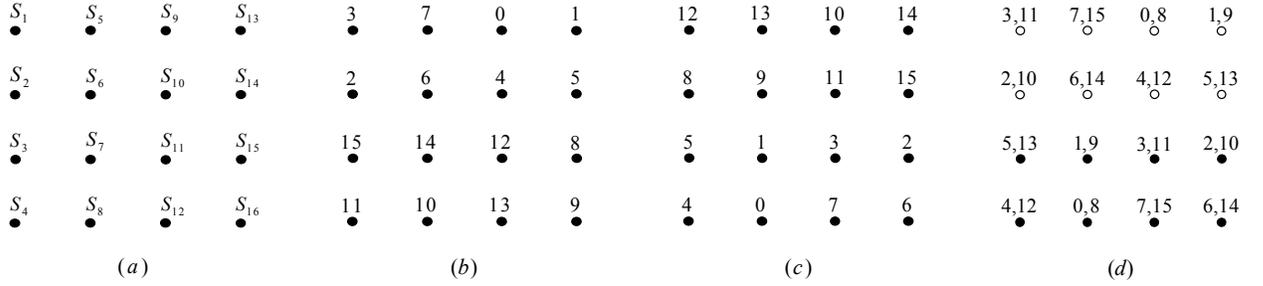}
\caption{(a). Symbol's arrangement in $16$-QAM, and achieved $16$-QAM mappings in decimal format: (b). $\lambda_{er}$, (c). $\lambda_{or}$, and (d). $\lambda_{el}$ (the unshaded symbols), $\lambda_{ol}$ (the shaded symbols).}
\label{16QAM}
\end{figure}

\begin{table}
\caption{Proposed $\lambda_{er}$, $\lambda_{or}$, $\lambda_{el}$,  and $\lambda_{ol}$ for $16$-QAM.}
\centering
\resizebox{0.8\columnwidth}{!}{%
\begin{tabular}{|c||cccccccccccccccc|}
\hline
$\lambda_{er}$ & [3 & 2 & 15 & 11 & 7 & 6 & 14 & 10 & 0 & 4 & 12 & 13 & 1 & 5 & 8 & 9] \\ \hline
$\lambda_{or}$ & [12 & 8 & 5 & 4 & 13 & 9 & 1 & 0 & 10 & 11 & 3 & 7 & 14 & 15 & 2 & 6] \\ \hline

$\lambda_{el}$ & [(3, & 11) & (2, & 10) & (7, & 15) & (6, & 14) & (0, & 8) & (4, & 12) & (1, & 9) & (5, & 13)]\\ \hline
$\lambda_{ol}$ & [(5, & 13) & (4, & 12) & (1, & 9) & (0, & 8) & (3, & 11) & (7, & 15) & (2, & 10) & (6, & 14)]\\\hline
\end{tabular}
\label{MD_16QAM_Map}
}
\end{table}

\begin{example}
\label{exmpl1}
\normalfont
 In the proposed MD mapping method, let us set $m=4$ ($16$-QAM), $N=3$ and $\mb{l} = (0,1,1,0,1,1,1,1,0,1,1,1)$. The label $\mb{l}$ is considered as a sequence of three $4$-bit labels, i.e., $\mb{l}= (\mb{l}_{1},\mb{l}_{2},\mb{l}_{3})$, where $\mb{l}_{1} = (0,1,1,0)$, $\mb{l}_{2} = (1,1,1,1)$, and $\mb{l}_{3} = (0,1,1,1)$. The mapping rule in (\ref{mapping_func}) is used to map $\mb{l}= (\mb{l}_{1},\mb{l}_{2},\mb{l}_{3})$ to signal point $\mb{x} = (x_{1}, x_{2}, x_{3})$, as follows. The Hamming weight of $\mb{l}$ is odd, i.e., $\mb{l} \in \mathcal{L}_{o}$, thus  $x_{1} = \lambda_{ol}(\mb{l}_{1})$, $x_{2} = \lambda_{or}(\mb{l}_{2})$, and $x_{3} = \lambda_{or}(\mb{l}_{3})$. The decimal format of $\mb{l}_{1}$, $\mb{l}_{2}$, and $\mb{l}_{3}$ are $6$, $15$, and $7$, respectively. In Fig. \ref{16QAM}(d), it can be observed that among the shaded symbols that $\lambda_{ol}$ operates on, symbol $S_{16}$ is mapped by decimal label $6$. As a result, $x_{1} = S_{16}$. Considering the mapping function $\lambda_{or}$ indicated in Fig. \ref{16QAM}(c) we also have $x_{2} = \lambda_{or}(15) = S_{14}$ and $x_{3} = \lambda_{or}(7) = S_{12}$. Consequently, $\mb{l}$ is mapped to $\mb{x} = (S_{16}, S_{14}, S_{12})$.

\end{example}

\subsection{Performance Comparison}
\color{black}The proposed mappings are compared to the MD mappings obtained using the state of the art methods in the literature such as the optimum mapping method in \cite{MD-8PSK-Ha}, the BSA, random mapping, and the MD mapping method in \cite{MD_16_64QAM}. To assess different mappings, we first compare the value of $\Phi(\mu,\mb{\chi})$ and $\hat{\Phi}(\mu,\mb{\chi})$ offered by the mappings. Then, we use the BER curve to compare the BICM-ID error performance in the low SNR region when using different mappings. Finally, we use an analytical bound on the error-floor to evaluate the system's error performance in the high SNR region. 
In our simulations, we consider a rate-$1/2$ convolutional code with the generator polynomial of $(13, 15)_{8}$.  An interleaver  length of about $10000$ bits is used. 
All BER curves are presented with seven iterations
and all gains reported in this section are  measured at a BER of $10^{-6}$. Also, the error-floor bounds have been plotted using the Gauss-Chebyshev method in \cite{Chindapol_16QAM}. \color{black} It is worth noting that achieving these error-floor bounds can be challenging because they happen at a very small value of BER. \color{black}
 
According to the modulation order, we discuss our results for three classes of modulations: low, medium, and high order modulations. In particular, we discuss $8$-ary modulations (for the low order), $16$-
and $32$-QAM (for the medium order), and M-QAM ($M = 64$, $128$ and $256$) for the higher order modulations.

\subsubsection{MD mapping of low order modulations}

\begin{table}[t]
\caption{Comparison of $\Phi(\mu,\mb{\chi})$ and $\hat{\Phi}(\mu,\mb{\chi})$ for different mappings.}
\centering
\resizebox{0.56\columnwidth}{!}{%
\begin{tabular}{|c||c|c||c|c|}
\hline
\multirow{2}{*}{Mapping}&\multicolumn{2}{c||}{$N = 2$}&\multicolumn{2}{c|}{$N = 3$}\\
\cline{2-5}
  & $\Phi(\mu,\mb{\chi})$ & $\hat{\Phi}(\mu,\mb{\chi})$ & $\Phi(\mu,\mb{\chi})$ & $\hat{\Phi}(\mu,\mb{\chi})$ \\
\hline\hline

 Optimum MD $8$-PSK \cite{MD-8PSK-Ha}  & 0.3112 &   3.3529 & 0.2119 & 3.5454 \\ \hline
 
Proposed MD $8$-PSK &  0.3112 &   3.3529 & 0.2119 & 3.5454 \\ \hline

BSA MD optimum $8$-QAM & 0.4506 & 2.7838 & 0.3002 & 2.8515 \\ \thickhline

Proposed MD optimum $8$-QAM  &  0.4552 & 2.9697 & 0.3041 & 3.1463 \\ \hline
\end{tabular}
\label{Parameters_8ary}
}
\end{table}

\begin{figure}[h]
\centering
\includegraphics[width= 0.65\columnwidth,viewport= 20mm 0mm 325mm 225mm,clip]{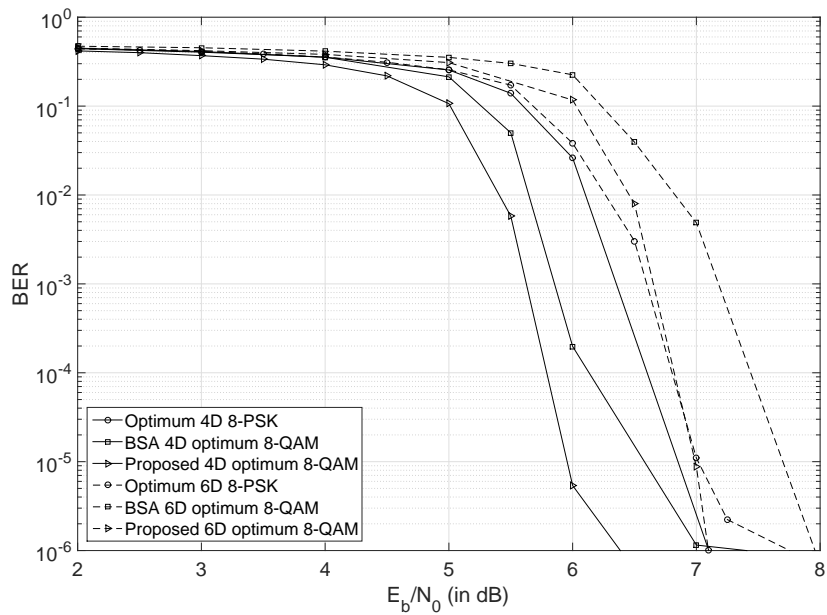}
\caption{BER performance of BICM-ID with $4$-D and $6$-D $8$-ary modulations over Rayleigh fading channels.}
\label{BER_4D_6D_8ary}
\end{figure}

Table \ref{Parameters_8ary} compares the values of the harmonic mean before and after feedback, i.e., $\Phi(\mu,\mb{\chi})$ and $\hat{\Phi}(\mu,\mb{\chi})$, for different $2N$-D ($N=2, 3$) mappings of $8$-ary modulations. As this table shows, for $8$-PSK, our proposed MD mapping provides the same values as those of the optimum MD mapping proposed in \cite{MD-8PSK-Ha}. \color{black} This clearly shows our proposed algorithm's efficiency.  We have also used our algorithm to obtain MD mappings of the optimum $8$-QAM constellation introduced in \cite{optimal_8QAM}. \color{black} As shown in Table \ref{Parameters_8ary}, for the optimum $8$-QAM,  our resulting mappings improve the values of $\Phi(\mu,\mb{\chi})$ and $\hat{\Phi}(\mu,\mb{\chi})$ in comparison with the mappings obtained using the state of the art BSA. Therefore, it is expected that for the optimum $8$-QAM, the proposed mappings outperform the BSA mappings in both the low and high SNR regions. This is confirmed by the plotted simulation results for the BER and error-floor bounds of BICM-ID in Fig. \ref{BER_4D_6D_8ary} and Fig. \ref{EF_4D_6D_8ary}, respectively. As shown in Fig. \ref{BER_4D_6D_8ary}, for the $4$-D and $6$-D optimum $8$-QAM, our proposed mappings outperform their BSA counterparts by $1$ and $0.85$ dB, respectively. The gain over the optimum $4$-D and $6$-D mapping of $8$-PSK is $0.7$ and $0.65$, respectively. This is  because in addition to a large value of $\hat{\Phi}(\mu,\mb{\chi})$, the proposed MD mapping of the optimum $8$-QAM significantly improves the value of $\Phi(\mu,\mb{\chi})$ compared to that of the optimum MD $8$-PSK. Note that for the case of MD $8$-PSK, our resulting mapping is the same as the optimum mapping in \cite{MD-8PSK-Ha}. Thus, only the  optimum mapping is used for comparison in BER and error-floor plots. As shown in Fig. \ref{EF_4D_6D_8ary}, the proposed $4$-D and $6$-D mappings of the optimum $8$-QAM outperform the corresponding BSA mappings by $0.3$ and $0.45$ dB, respectively. Moreover in this figure, the optimum MD mapping of $8$-PSK offers a lower error-floor due to the nature of the $8$-PSK constellation, which can offer higher values of $\hat{\Phi}(\mu,\mb{\chi})$.

\begin{figure}[t]
\centering
\includegraphics[width= 0.68\columnwidth,viewport= 10mm 0mm 325mm 225mm,clip]{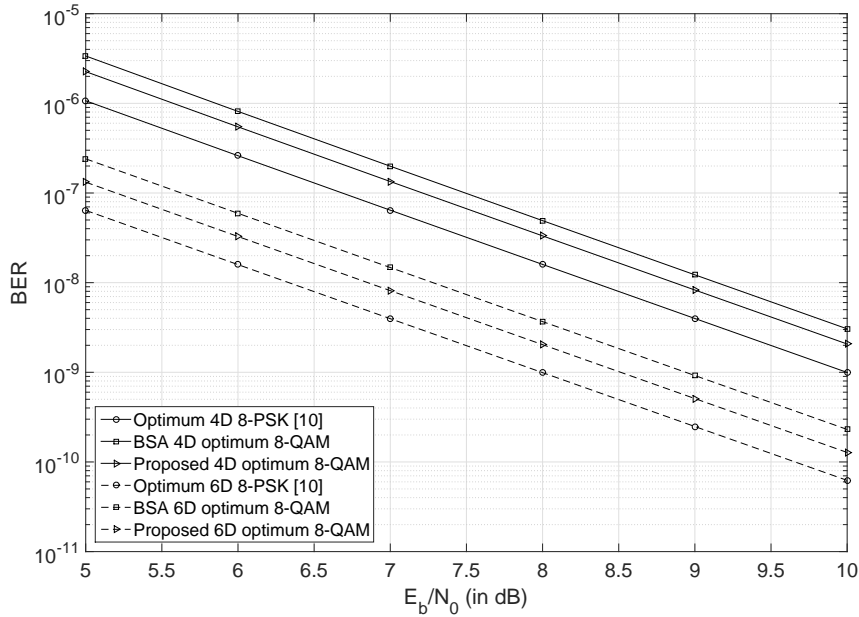}
\caption{Error-floor bounds of BER for BICM-ID with $4$-D and $6$-D $8$-ary modulations over Rayleigh fading channels.}
\label{EF_4D_6D_8ary}
\end{figure}

\color{black}

\subsubsection{MD mapping of medium order modulations}

\begin{table}[b]
\caption{Comparison of $\Phi(\mu,\mb{\chi})$ and $\hat{\Phi}(\mu,\mb{\chi})$ for different mappings.}
\centering
\resizebox{0.5\columnwidth}{!}{%
\begin{tabular}{|c|c||c|c|}
\hline
\multicolumn{2}{|c||}{$4$D Mapping} & $\Phi(\mu,\mb{\chi})$ & $\hat{\Phi}(\mu,\mb{\chi})$ \\
\hline\hline
\multirow{4}{*}{$16$-QAM} & BSA mapping &  0.2026 & 2.5814 \\ \cline{2-4}
&Random mapping & 0.2012   & 1.4350   \\ \cline{2-4}
&Mapping in \cite{MD_16_64QAM} & 0.2151  & 2.8491 \\ \cline{2-4}
&Proposed mapping & 0.2151  & 3.1622  \\ \thickhline
\multirow{3}{*}{$32$-QAM} & BSA mapping &  0.1027 & 2.8574 \\ \cline{2-4}
&Random mapping & 0.1008  & 1.2567 \\ \cline{2-4}
&Proposed mapping & 0.1117  &  3.1677  \\ \hline
\end{tabular}
\label{Parameters_16_32QAM}
}
\end{table}

\begin{figure}[h]
\centering
\includegraphics[width= 0.63\columnwidth,viewport= 20mm 0mm 335mm 225mm,clip]{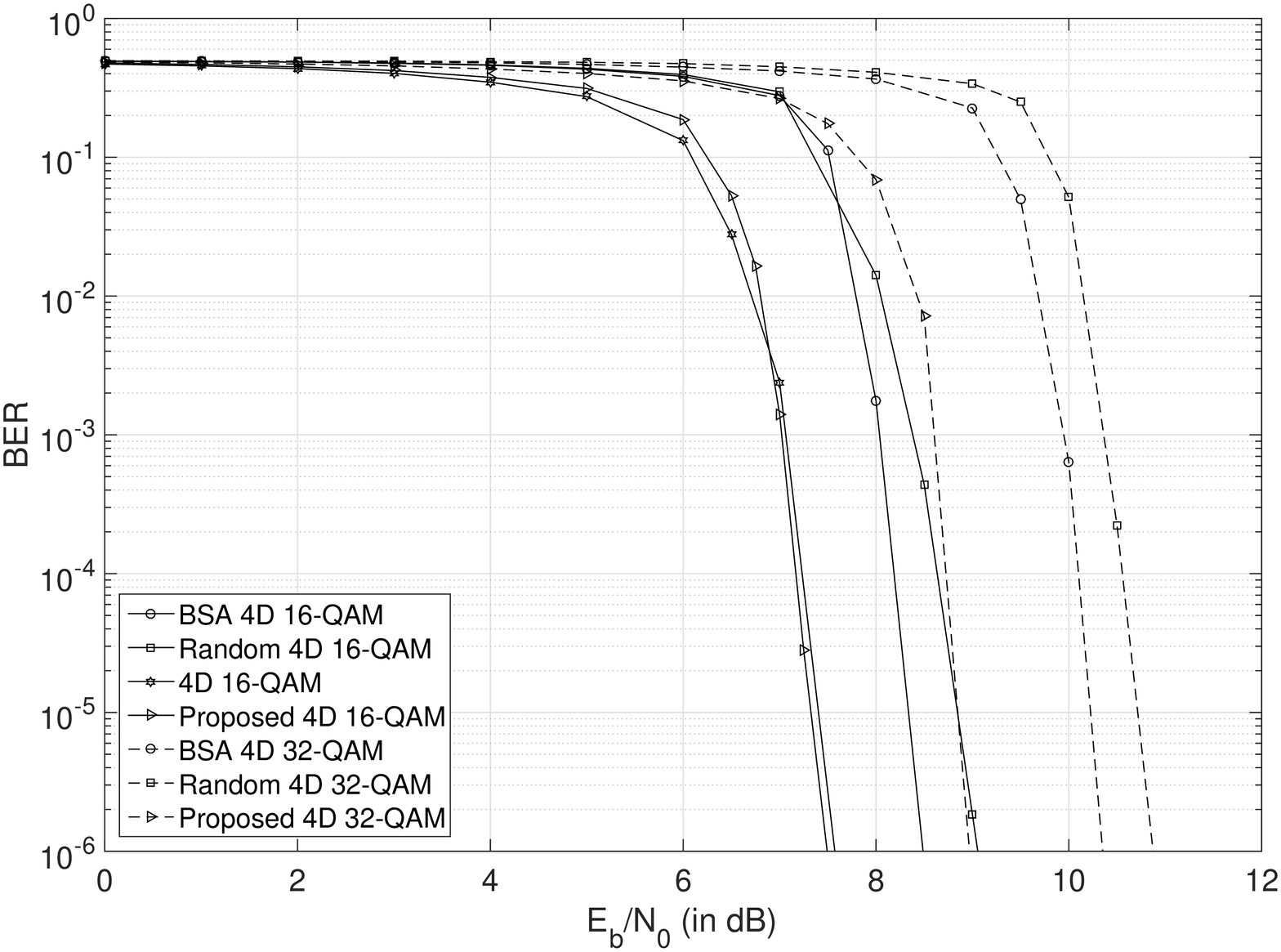}
\caption{BER performance of BICM-ID with $4$-D $16$- and $32$-QAM over Rayleigh fading channels.}
\label{BER_4D_16_32QAM}
\end{figure}

\begin{figure}[h]
\centering
\includegraphics[width= 0.63\columnwidth,viewport= 15mm 0mm 335mm 225mm,clip]{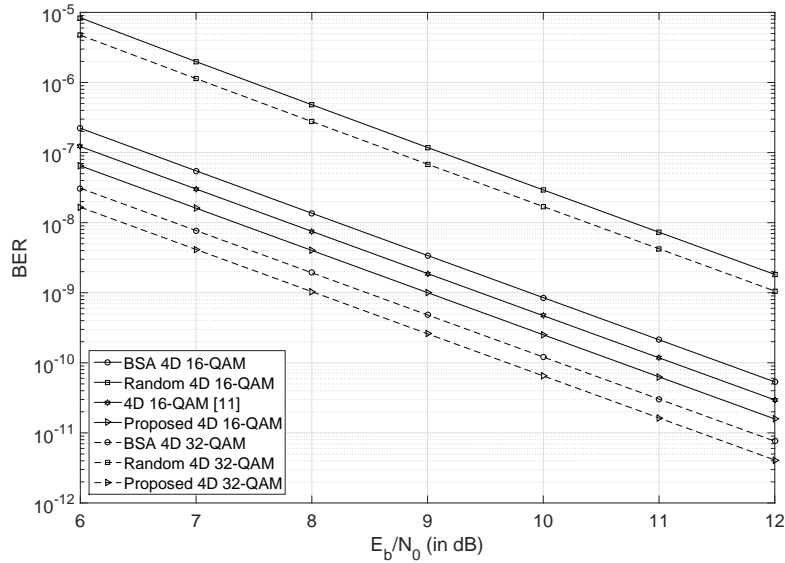}
\caption{Error-floor bounds of BER for BICM-ID with $4$-D  $16$- and $32$-QAM over Rayleigh fading channels.}
\label{EF_4D_16_32QAM}
\end{figure}

Table \ref{Parameters_16_32QAM} lists the values of $\Phi(\mu,\mb{\chi})$ and $\hat{\Phi}(\mu,\mb{\chi})$ for different $4$-D mappings of $16$- and $32$-QAM. It can be seen from this table that, except for the value of $\Phi(\mu,\mb{\chi})$ for the $4$-D $16$-QAM where the proposed mapping offers the same value as that of the mapping in \cite{MD_16_64QAM}, our proposed mappings significantly improve the values of $\Phi(\mu,\mb{\chi})$ and $\hat{\Phi}(\mu,\mb{\chi})$, compared to their well-known counterparts. Therefore, it is expected that the proposed mappings improve the error rate performance compared to  the previously known mappings in both the low and high SNR regions. This is confirmed by the simulation results for the BER performance of BICM-ID shown in Fig. \ref{BER_4D_16_32QAM} and by the error-floor bounds plotted in Fig. \ref{EF_4D_16_32QAM}. In the case of $4$-D $16$-QAM, as illustrated in Fig. \ref{BER_4D_16_32QAM}, in the low SNR region, the proposed mapping achieves a gain of $1$ and $1.55$ dB over the BSA and random mappings, respectively. Although as shown in Fig. \ref{BER_4D_16_32QAM} the BER performance of the proposed $4$-D $16$-QAM is similar to that of the mapping in \cite{MD_16_64QAM}, our mapping outperforms the mapping in \cite{MD_16_64QAM} in the high SNR region by about $0.5$ dB, as shown in Fig. \ref{EF_4D_16_32QAM}. This figure also shows that the proposed mapping improves the error-floor by $0.9$ and $3.4$ dB compared to the BSA and random mappings, respectively. In the case of $4$-D $32$-QAM,  our proposed mapping outperforms the best previously known mappings, i.e., the BSA and random mappings, by $1.4$ and $1.9$ dB, respectively, which is illustrated in Fig. \ref{BER_4D_16_32QAM}. Fig. \ref{EF_4D_16_32QAM} shows that the corresponding gain in the high SNR region is $0.5$ and $4$ dB, respectively.

\subsubsection{MD mapping of higher order modulations}

\begin{table}[b]
\caption{Comparison of $\Phi(\mu,\mb{\chi})$ and $\hat{\Phi}(\mu,\mb{\chi})$ for different mappings.}
\centering
\resizebox{0.5\columnwidth}{!}{%
\begin{tabular}{|c|c||c|c|}
\hline
\multicolumn{2}{|c||}{$4$D Mapping} & $\Phi(\mu,\mb{\chi})$ & $\hat{\Phi}(\mu,\mb{\chi})$ \\
\hline\hline
\multirow{4}{*}{$64$-QAM} & BSA mapping &  0.0481 & 2.6899 \\ \cline{2-4}
&Random mapping & 0.0478  & 1.1688   \\ \cline{2-4}
&Mapping in \cite{MD_16_64QAM} & 0.0579  & 2.8166 \\ \cline{2-4}
& Proposed mapping & 0.0568 & 3.1683  \\ \thickhline

\multirow{3}{*}{$128$-QAM} & BSA mapping &  0.0245 & 1.8566 \\ \cline{2-4}
&Random mapping & 0.0245  & 1.1430 \\ \cline{2-4}
&Proposed mapping & 0.0294 & 3.2273 \\ \thickhline

\multirow{3}{*}{$256$-QAM} & BSA mapping &  0.0118 & 1.1282 \\ \cline{2-4}
&Random mapping & 0.0120 & 1.1034 \\ \cline{2-4}
&Proposed mapping & 0.0144 & 3.2389 \\ \hline

\end{tabular}
\label{Parameters_higher_order}
}
\end{table}

\begin{figure}[t]
\centering
\includegraphics[width= 0.75\columnwidth,viewport= 10mm 0mm 355mm 225mm,clip]{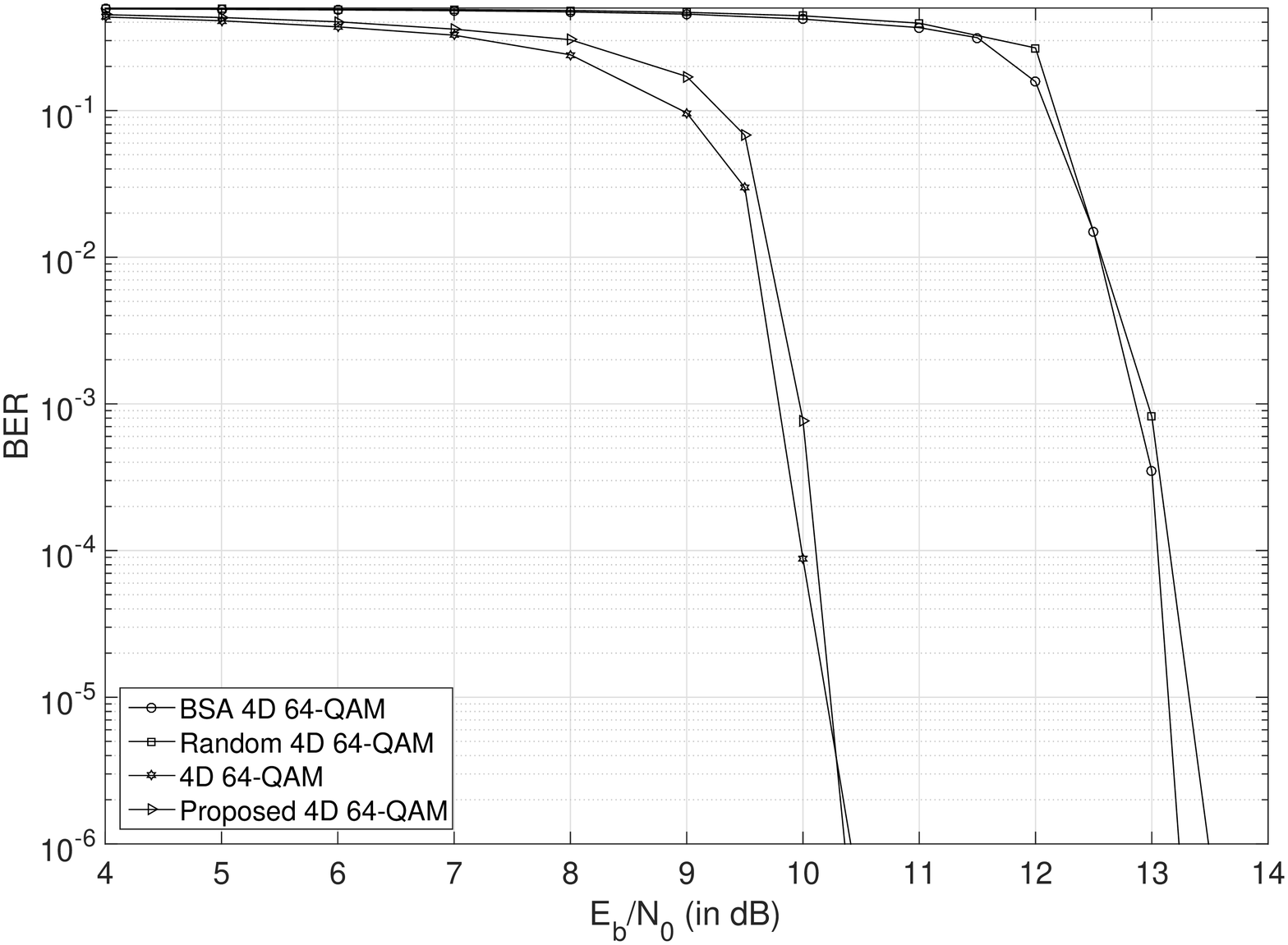}
\caption{BER performance of BICM-ID with $4$-D $64$-QAM over Rayleigh fading channels.}
\label{BER_4D_64QAM}
\end{figure}

\begin{figure}[h]
\centering
\includegraphics[width= 0.75\columnwidth,viewport= 10mm 0mm 355mm 225mm,clip]{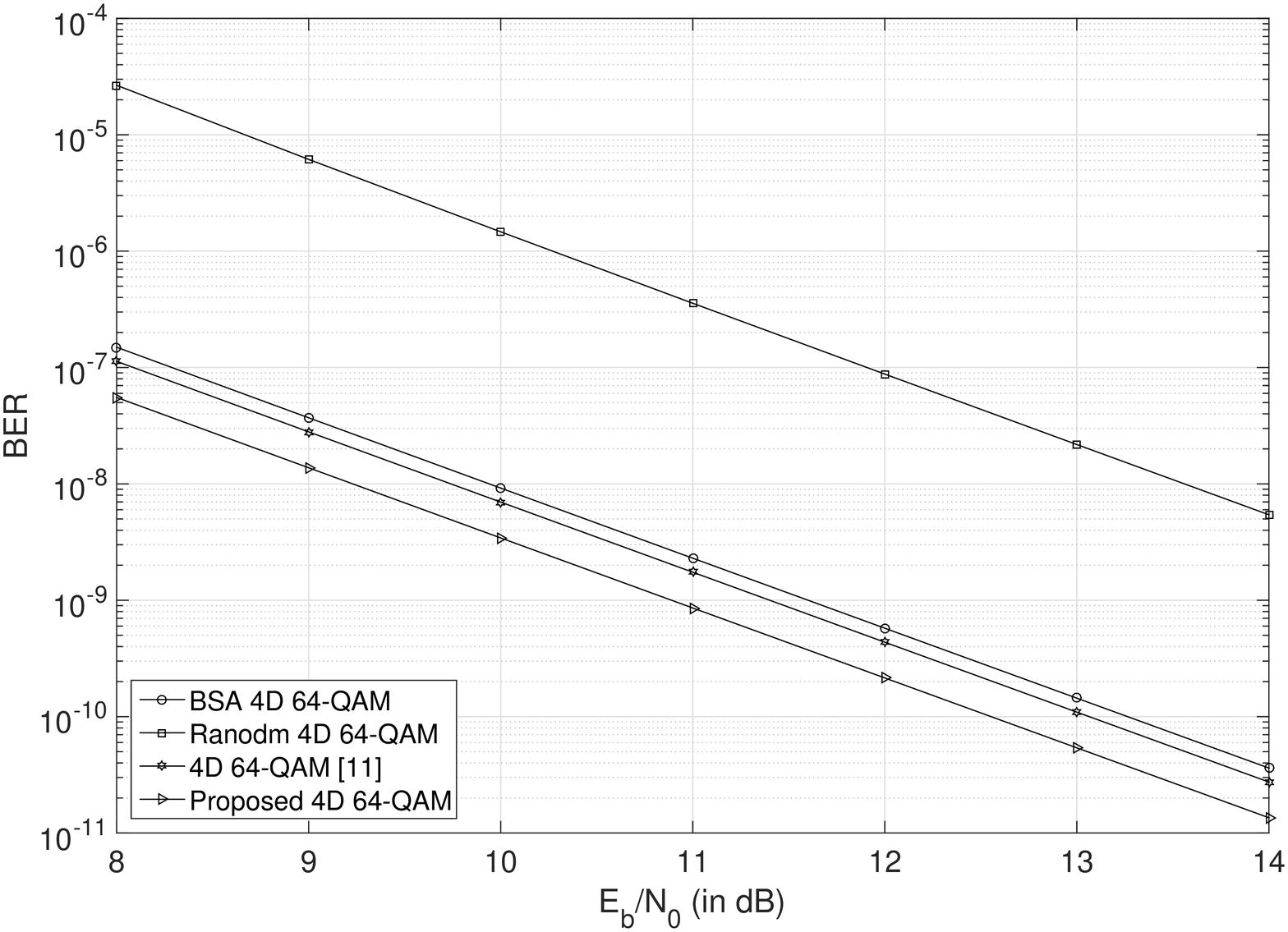}
\caption{Error-floor bounds of BER for BICM-ID with $4$-D  $64$-QAM over Rayleigh fading channels.}
\label{EF_4D_64QAM}
\end{figure}

Table \ref{Parameters_higher_order} reports the values of  $\Phi(\mu,\mb{\chi})$ and $\hat{\Phi}(\mu,\mb{\chi})$ for  MD mappings of different higher order modulations.  It can be observed from this table that our proposed mappings improve these values compared to the previously known mappings, except for the case of the MD $64$-QAM in \cite{MD_16_64QAM}, where our mapping offers a similar value of  $\Phi(\mu,\mb{\chi})$. 
Therefore, it is expected that the proposed mappings offer improved error performance in both the low and high SNR regions. This is confirmed by the simulation results for the BER and by the analytical error bounds in Fig. \ref{BER_4D_64QAM}-\ref{EF_4D_128_256QAM}. In the case of $4$-D $64$-QAM, our proposed mapping outperforms the BSA and random mappings in the low SNR region by $2.9$ and $3.1$ dB, respectively, as shown in Fig. \ref{BER_4D_64QAM}. Although in Fig. \ref{BER_4D_64QAM} our proposed mapping offers similar performance to that of the mapping in \cite{MD_16_64QAM},  the proposed mapping outperforms the mapping in \cite{MD_16_64QAM} by over $0.5$ dB in the high SNR region, as shown in Fig. \ref{EF_4D_64QAM}. In this figure, the achieved gain over the BSA and random mappings is $0.75$ and $4.3$ dB, respectively. As shown in Fig. \ref{BER_4D_128_256QAM}, for $4$-D $128$-QAM in the low SNR region, the proposed mapping achieves the gain of $3.1$ and $3$ dB over the BSA and random mappings, respectively. The corresponding gain for the case of $4$-D $256$-QAM is $3.5$ and $3.2$ dB, respectively. For the case of $4$-D $128$-QAM in the high SNR region, our proposed mapping improves the error-floor by $2$ and $4.5$ dB compared to the BSA and random mappings, respectively, as illustrated in Fig. \ref{EF_4D_128_256QAM}. The corresponding gain for the case of $4$-D $256$-QAM is $4.5$ and $4.6$ dB, respectively. 

\FloatBarrier
\begin{figure}[H]
\centering
\includegraphics[width= 0.75\columnwidth,viewport= 20mm 0mm 365mm 225mm,clip]{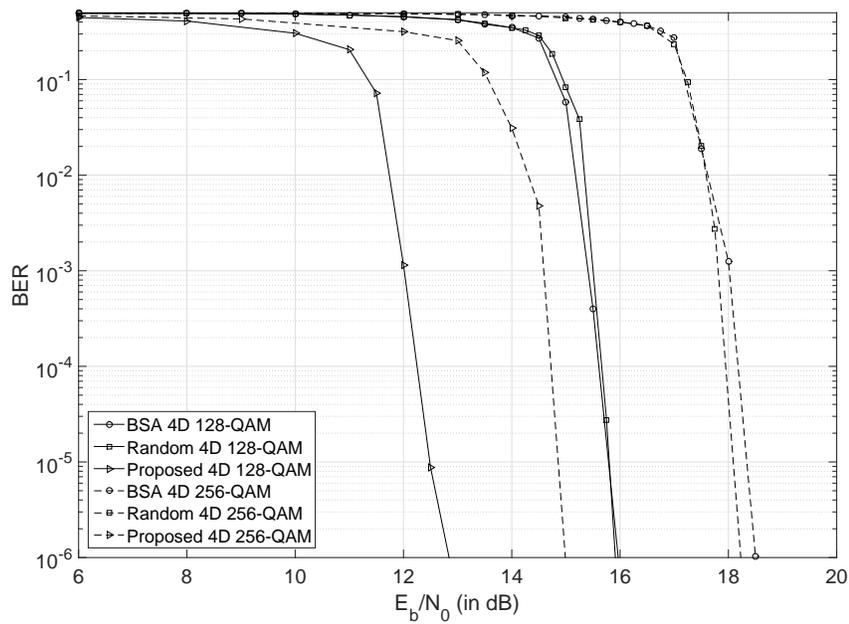}
\caption{BER performance of BICM-ID with $4$-D $128$- and $256$-QAM over Rayleigh fading channels.}
\label{BER_4D_128_256QAM}
\end{figure}

\FloatBarrier
\begin{figure}[h!]
\centering
\includegraphics[width= 0.75\columnwidth,viewport= 18mm 0mm 365mm 225mm,clip]{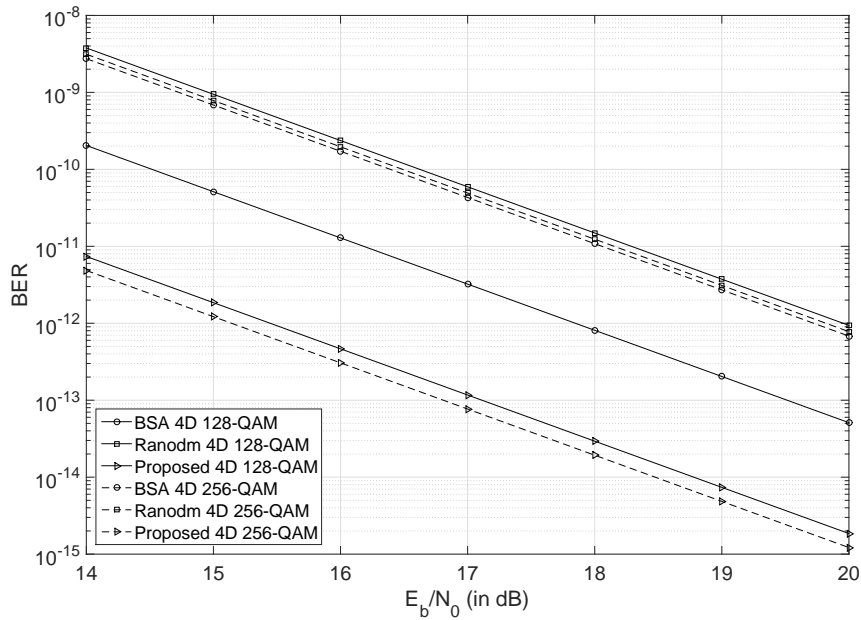}
\caption{Error-floor bounds of BER for BICM-ID with $4$-D $128$- and $256$-QAM over Rayleigh fading channels.}
\label{EF_4D_128_256QAM}
\end{figure}

It can be observed from Fig. \ref{BER_4D_6D_8ary}-\ref{EF_4D_128_256QAM} that the BSA results become less efficient  as the modulation order increases. As such,  random mapping performs better than the BSA mapping  for $4$-D $256$-QAM, as shown in Fig. \ref{BER_4D_128_256QAM}.  

Fig. \ref{BER_LDPC_BICMID_1024} plots the BER curves for different MD mappings of $1024$-QAM. As it can be seen from this figure, the proposed mapping significantly outperforms the other counterparts. It is also worth noting that in this figure, the random mapping outperforms  the BSA mapping. This is because  when looking for a mapping for $4$-D $1024$-QAM, the BSA is unable to finish one round of the algorithm in a reasonable time frame. As a result, the obtained BSA  mapping is not very suitable.

\begin{figure}[H]
\begin{center}
\includegraphics[width=.7\columnwidth, viewport= 10mm 0mm 325mm 225mm,clip]{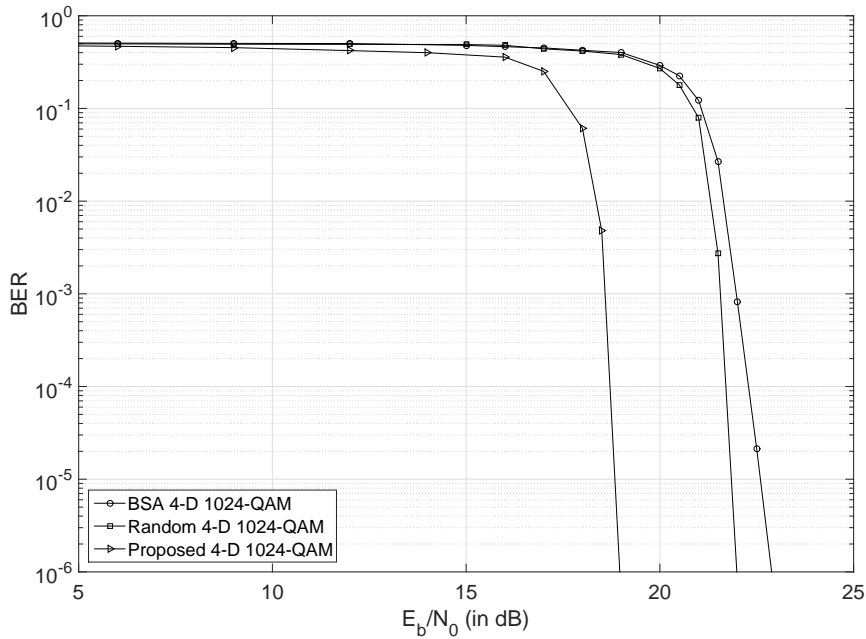} \caption{BER performance of BICM-ID with $4$-D $1024$-QAM over Rayleigh fading channels.} 
\label{BER_LDPC_BICMID_1024}
\end{center}
\end{figure}

\subsection{Analysis of Convergence Behaviour} 
The extrinsic information transfer chart (EXIT chart) \cite{EXIT} is a commonly used  metric to assess the convergence behavior of BICM-ID. In an EXIT chart, the area between the decoder curve and the demapper curve is called EXIT tunnel \cite{Exit1}. BICM-ID can achieve a coding gain through the iterative decoding process only when it provides an open EXIT tunnel. Fig. \ref{EXIT_Rayleigh} shows the EXIT charts for BICM-ID when using different $4$-D mappings of $128$-QAM  in the Rayleigh fading channel. For brevity in this section, we investigate the EXIT chart only for $4$-D  $128$-QAM. The results are similar for other considered modulations.   
From Fig. \ref{EXIT_Rayleigh}, it can be seen that  BICM-ID with the proposed mapping exhibits  an open Exit tunnel when $\frac{E_b}{N_0} = 6$ dB. This implies that when using the proposed mapping, the iterative decoding process starts to improve the performance of BICM-ID at $\frac{E_b}{N_0} = 6$. This is in accordance with the BER curve of the proposed mapping in Fig. \ref{BER_4D_128_256QAM}, where the BER curve falls 
gradually after $\frac{E_b}{N_0} = 6$. 
 Fig. \ref{EXIT_Rayleigh} also shows that the open EXIT tunnels for the BSA and random mappings appear at $\frac{E_b}{N_0} =11$. As a result,  when using the BSA and random mappings, the system BER will not improve unless after $\frac{E_b}{N_0} =11$, which is confirmed by the corresponding BER curves in Fig. \ref{BER_4D_128_256QAM}.
Consequently, the BER performance with our proposed mapping   improves through the iterative decoding $5$ dB earlier than those of the BSA and Random mappings.  
This results in an earlier turbo cliff for the BER curve with our proposed mapping in Fig. \ref{BER_4D_128_256QAM}. 
 
\begin{figure}[h!]
\centering
\includegraphics[width= 0.7\columnwidth,viewport= 20mm 0mm 320mm 225mm,clip]{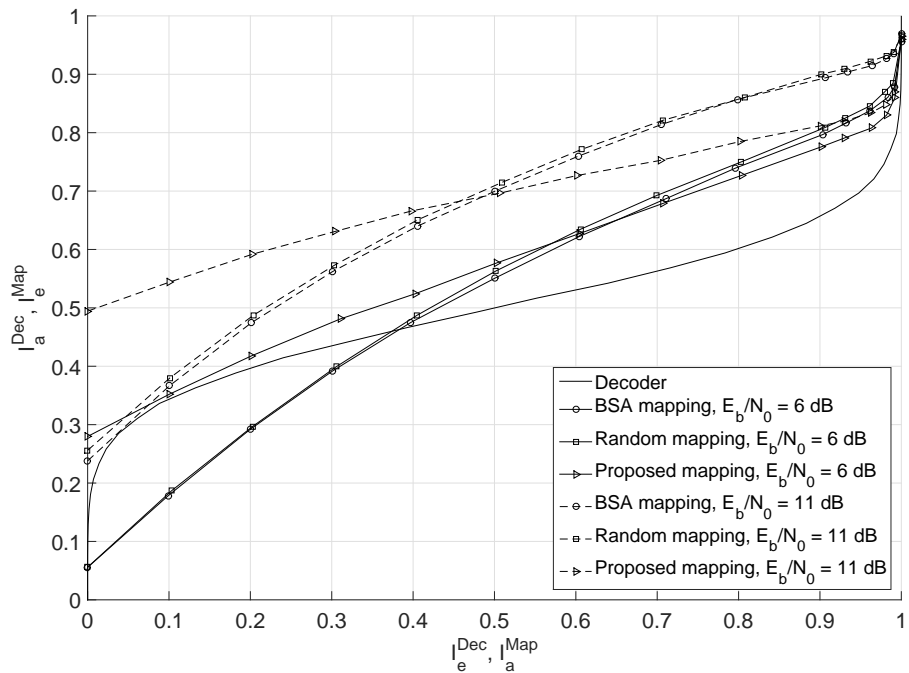}
\caption{EXIT chart for different $4$-D mappings of $128$-QAM in the Rayleigh fading channel.}
\label{EXIT_Rayleigh}
\end{figure}

\section{Conclusion}
\label{conc.}
A novel MD mapping method is proposed to improve the error performance of the BICM-ID system in both the low and high SNR regions over Rayleigh fading channels. The method uses four $2$-D mappings to construct an MD mapping that improves the error performance in the low SNR region. Furthermore, cost functions are developed and optimized over the $2$-D mappings to  achieve an MD mapping that improves the error performance in the high SNR region. Due to the lower complexity of the $2$-D space, the optimization approach is very simple and results in excellent MD mappings for different modulations, including higher order modulations such as $2^{m}$-QAM ($m = 4,..,10$). Extensive numerical results, including analytical and simulation results, show that the obtained mappings significantly outperform the previously known state of the art mappings in both the low and high SNR regions.  

\appendices
\section{Proof of Proposition \ref{prop1}}
\label{apdx1}

Let $\mb{i} = (i_1, i_2, \cdots, i_D)$ be the set of bit positions at which $\mb{l}_1$ and $\mb{l}_2$  are different. Let us define $o_{j}$ and $z_{j}$ as

\begin{eqnarray}
\label{prop0_1}
o_{j} = \sum_{k=1}^{D} I(l_{j}^{i_{k}} =1), \\ \nonumber
z_{j} = \sum_{k=1}^{D} I(l_{j}^{i_{k}} =0),
\end{eqnarray} where $j=1,2$ and $I(\theta)$ is an indicator function that takes value one if $\theta$ is true, otherwise it equals zero. Clearly, $z_{1} = o_2$ and $z_2 = o_1$. Let  $\bar{o}$ be the number of bit positions in which both $\mb{l}_1$ and $\mb{l}_2$ represent the bit value one. As $w_1 = o_1+\bar{o}$ and $w_2 = o_2+\bar{o}$, we can write 
\begin{eqnarray}
\label{prop0_2}
W = w_1+w_2 &=& (o_1+\bar{o}) + (o_2+\bar{o}),\\ \nonumber
 &=& o_1 + o_2 + 2\bar{o},\\ \nonumber
 & \overset{(o_2 = z_1)}{=} & o_1 + z_1 + 2\bar{o},\\ \nonumber
 &=& D + 2\bar{o}.
\end{eqnarray} In (\ref{prop0_2}), $2\bar{o}$ is an even number; as a result,  if $W \in \mathbb{E}$, $D \in \mathbb{E}$, and if  $W \in \mathbb{O}$, $D \in \mathbb{O}$.

\section{Proof of Proposition \ref{prop2}}
\label{apdx2}

Let $\mb{l}$ and $\tilde{\mb{l}}$ be two $mN$-bit labels. If $\mb{l}, \tilde{\mb{l}} \in \mathcal{L}_{e}$ or $\mb{l}, \tilde{\mb{l}} \in \mathcal{L}_{o}$, the corresponding value for $W$ in Proposition \ref{prop1} is an even number. As a result, $\mb{l}$ and $\tilde{\mb{l}}$ 
have an even Hamming distance from each other. \color{black}
Similarly, two labels one from $\mathcal{L}_{e}$ and the other from $\mathcal{L}_{o}$  have an odd Hamming distance from each other. Since there is no common $2$-D signal point between $\chi_{el}$ and $\chi_{ol}$, there is no common $2N$-D signal point between $\mb{\chi}_{e}$ and $\mb{\chi}_{o}$. As a result, none of the $2N$-D signal points will be mapped simultaneously by a label from $\mathcal{L}_{e}$ and a label from $\mathcal{L}_{o}$. Therefore, it is sufficient to prove that there is a one-to-one correspondence between labels from $\mathcal{L}_{e}$ and signal points from $\mb{\chi}_{e}$ and similarly between labels in $\mathcal{L}_{o}$ and signal points in $\mb{\chi}_{o}$. In what follows, we prove this for the even subsets, i.e., for labels in $\mathcal{L}_{e}$ and signal points in $\mb{\chi}_{e}$.

Assume that $\mb{l} = (l_{1}, l_{2}, \cdots, l_{mN})$ and $\tilde{\mb{l}} = (\tilde{l}_{1}, \tilde{l}_{2}, \cdots, \tilde{l}_{mN})$ are two labels in $\mathcal{L}_{e}$ and are mapped to  $\mb{x} = (x_{1}, \cdots, x_{N})$ and $\tilde{\mb{x}} = (\tilde{x}_{1}, \cdots, \tilde{x}_{N})$, respectively, where both $\mb{x}$ and $\tilde{\mb{x}}$ are in $\mb{\chi}_{e}$. Let us define $\mb{l}_{i}$ and $\tilde{\mb{l}}_{i}$ as the $i^{th}$ $m$-tuple bits of $\mb{l}$ and $\tilde{\mb{l}}$, respectively. Then $\mb{l} = (\mb{l}_{1}, \mb{l}_{2}, \cdots, \mb{l}_{N})$ and $\tilde{\mb{l}} = (\tilde{\mb{l}}_{1}, \tilde{\mb{l}}_{2}, \cdots, \tilde{\mb{l}}_{N})$. Based on the relation between $\mb{l}_{i}$ and $\tilde{\mb{l}}_{i}$ for different values of $i$, there are two possible cases as follows:

\emph{Case 1:} There exists a value of $i$ ($i\geqslant 2$) such that $\mb{l}_{i} \neq \tilde{\mb{l}}_{i}$.
Let $j\geqslant 2$, then according to (\ref{mapping_func}), the same one-to-one mapping function, i.e., $\lambda_{er}(.)$, is used to map $\mb{l}_{j}$ to $x_{j}$ and $\tilde{\mb{l}}_{j}$ to $\tilde{x}_{j}$. Therefore, because $\mb{l}_{j} \neq \tilde{\mb{l}}_{j}$, we have 

\begin{equation}
\label{prop1_case1}
x_{j} \neq \tilde{x}_{j} \Rightarrow \mb{x} \neq \tilde{\mb{x}}.
\end{equation}


\emph{Case 2:} $\mb{l}_{i}$ = $\tilde{\mb{l}_{i}}$ for all $i \geqslant 2$.
In this case, $\mb{l}_{i} \neq \tilde{\mb{l}_{i}}$ only when $i=1$, and as a result, \color{black} $d_{H}(\mb{l},\tilde{\mb{l}}) = d_{H}(\mb{l}_{1},\tilde{\mb{l}}_{1})$.  \color{black} Since $\mb{l}$ and $\tilde{\mb{l}}$ belong to $\mathcal{L}_{e}$, they have an even Hamming distance from each other. Consequently, the Hamming distance between $\mb{l}_{1}$ and $\tilde{\mb{l}}_{1}$ is even as well. However, the two labels that are mapped to each symbol in $\chi_{el}$ have an odd Hamming distance from each other. Therefore, because $ \lambda_{el}(\mb{l}_{1}) \neq \lambda_{el}(\tilde{\mb{l}}_{1})$, we have

\begin{eqnarray}
\label{prop1_case2}
 x_{1} \neq \tilde{x}_{1} 
\Rightarrow \mb{x} \neq \tilde{\mb{x}}.
\end{eqnarray}
From (\ref{prop1_case1}) and (\ref{prop1_case2}), it is concluded that in the proposed mapping function, different labels from $\mathcal{L}_{e}$ are mapped to different signal points in $\mb{\chi}_{e}$. In a similar way, it can be proven  that the different labels from $\mathcal{L}_{o}$ are mapped to the different signal points in $\mb{\chi}_{o}$. As a result, the proposed MD mapping function is bijective.


\section{Proof of Proposition \ref{prop3}}
\label{apdx3}
Similar to proposition \ref{prop1}, assume that $\mb{l} = (\mb{l}_{1}, \mb{l}_{2}, \cdots, \mb{l}_{N})$ and $\mb{l}^{\prime} = (\mb{l}^{\prime}_{1}, \mb{l}^{\prime}_{2}, \cdots, \mb{l}^{\prime}_{N})$  are two labels in $\mathcal{L}_{e}$ and are mapped to  $\mb{x} = (x_{1}, \cdots, x_{N})$ and $\mb{x}^{\prime} = (x_{1}^{\prime}, \cdots, x_{N}^{\prime})$, respectively, where both $\mb{x}$ and $\mb{x}^{\prime}$ are in $\mb{\chi}_{e}$, and   $\mb{l}_{i}$ and $\mb{l}^{\prime}_{i}$ are the $i^{th}$ $m$-tuple bits of $\mb{l}$ and $\mb{l}^{\prime}$, respectively. Let us assume that  $d_{H}(\mb{l},\mb{l}^{\prime}) \geq m+2$  and $j$ is the number of values for $i$ such that $\mb{l}_{i} \neq \mb{l}^{\prime}_{i}$. Note that $d_{H}(\mb{l},\mb{l}^{\prime}) > m$, thus $j>1$. Based on the value of $j$, there are two possible cases as follows. 

\emph{Case 1 (when $j=2$):} Assume that $\mb{l}_{i}\neq\mb{l}^{\prime}_{i}$ for $i = p, q$, where $q>p \geqslant 1$. We have 

\begin{eqnarray}
\label{prop3_case1}
 d_{H}(\mb{l},\mb{l}^{\prime}) =  d_{H}(\mb{l}_{p},\mb{l}^{\prime}_{p}) +d_{H}(\mb{l}_{q},\mb{l}^{\prime}_{q}).
\end{eqnarray} Let $p=1$. Because $\mb{l}_{q}$ and $\mb{l}^{\prime}_{q}$ are $m$-tuple vectors, $d_{H}(\mb{l}_{q},\mb{l}^{\prime}_{q}) \leq m$. Moreover, $d_{H}(\mb{l},\mb{l}^{\prime}) \geq m+2$. Thus, using (\ref{prop3_case1}) we have $d_{H}(\mb{l}_{p},\mb{l}^{\prime}_{p}) \geq 2$. Since $p=1$, the mapping function  $\lambda_{el}$ in (\ref{mapping_func}) is used to map $\mb{l}_{p}$ and $\mb{l}^{\prime}_{p}$. Note that in $\lambda_{el}$, two $m$-tuple labels with Hamming distance more than one bit cannot be mapped to the same $2$-D symbol. As a result, we have $\lambda_{el}(\mb{l}_{p})\neq\lambda_{el}(\mb{l}^{\prime}_{p})$, and therefore, $x_{p} \neq x_{p}^{\prime}$. Similarly, since $q > 1$, the mapping function  $\lambda_{er}$ in (\ref{mapping_func}) is used to map $\mb{l}_{q}$ and $\mb{l}^{\prime}_{q}$. Moreover, $\lambda_{er}$ in (\ref{mapping_func}) maps different $m$-tuple labels to different $2$-D symbols. Therefore, as $\mb{l}_{q} \neq \mb{l}^{\prime}_{q}$,  $\lambda_{er}(\mb{l}_{q})\neq\lambda_{er}(\mb{l}^{\prime}_{q})$, and as a result, $x_{q} \neq x_{q}^{\prime}$. Let $p > 1$. In this case, $\lambda_{er}$ is used to map $\mb{l}_{p}$, $\mb{l}^{\prime}_{p}$, $\mb{l}_{q}$, and $\mb{l}^{\prime}_{q}$. As $\mb{l}_{p} \neq \mb{l}^{\prime}_{p}$  and $\mb{l}_{q} \neq \mb{l}^{\prime}_{q}$,  $\lambda_{er}(\mb{l}_{p})\neq\lambda_{er}(\mb{l}^{\prime}_{p})$ and $\lambda_{er}(\mb{l}_{q})\neq\lambda_{er}(\mb{l}^{\prime}_{q})$. As a result, $x_{p} \neq x_{p}^{\prime}$ and $x_{q} \neq x_{q}^{\prime}$. Therefore, for all values of $p$ and $q$, $\mb{x}$ and $\mb{x}^{\prime}$ are different in more than one symbol. 

\emph{Case 2 (when $j\geqslant 3$):} There are at least two values, $p$ and $q$, such that $p>q > 1$ and $\mb{l}_{i}\neq\mb{l}^{\prime}_{i}$ when $i = p, q$. In this case, $\lambda_{er}$ is used to map $\mb{l}_{p}$, $\mb{l}^{\prime}_{p}$, $\mb{l}_{q}$, and $\mb{l}^{\prime}_{q}$. As $\mb{l}_{p} \neq \mb{l}^{\prime}_{p}$  and $\mb{l}_{q} \neq \mb{l}^{\prime}_{q}$,  $\lambda_{er}(\mb{l}_{p})\neq\lambda_{er}(\mb{l}^{\prime}_{p})$ and $\lambda_{er}(\mb{l}_{q})\neq\lambda_{er}(\mb{l}^{\prime}_{q})$. 
As a result, $x_{i} \neq x_{i}^{\prime}$ when $i = p,q$. Therefore, $\mb{x}$ and $\mb{x}^{\prime}$ are different in more than one symbol. 

Consequently,  in the proposed MD mapping function, when the Hamming distance between two labels in $\mathcal{L}_{e}$ is larger than ($m+1$) bits, the corresponding symbols-vectors in  $\mb{\chi}_{e}$ are different in more than one symbol, and therefore, cannot be the nearest neighbours. The same characteristic can be proven for the labels in $\mathcal{L}_{o}$ and the corresponding symbol-vectors in $\mb{\chi}_{o}$.

\section{Proof of Proposition \ref{prop4}}
\label{apdx4}
Since $y_{i} \geqslant 0$ for all $i$, then 

\begin{equation}
 \frac{1}{\sum_{i = 1}^{N} y_{i}} \leqslant \frac{1}{y_{j}}, ~~~ j = 1, 2, \cdots, N.
\label{Prop1_Inq2}
\end{equation} By taking the summation over all values of $j$ from both sides of (\ref{Prop1_Inq2}), the inequality in (\ref{Prop1_Inq1}) can be written as 

\begin{eqnarray}
\label{Prop1_Inq3}
\frac{1}{\sum_{i = 1}^{N} y_{i}} \leqslant \frac{1}{N} \sum_{j=1}^{N} \frac{1}{y_{j}}.
\end{eqnarray}

\section{Proposed mappings}
\label{app_mappings}

\begin{table*} 
\caption{Proposed $\lambda_{er}$, $\lambda_{or}$, $\lambda_{el}$,  and $\lambda_{ol}$ for $32$-QAM.}
\centering
\resizebox{.8\columnwidth}{!}{%
\begin{tabular}{|c||cccccccccccccccc|}
\hline
\multirow{2}{*}{$\lambda_{er}$} & [24 & 17 & 13 & 8 & 25 & 29 & 1 & 9 & 21 & 28 & 0 & 5 & 16 & 20 & 4 & \multicolumn{1}{c|}{12,}\\ &30& 22 & 6 & 2 & 23 & 18 & 14 & 7 & 26 & 19 & 15 & 11 & 27 & 31 & 3 & 10] \\ \hline
\multirow{2}{*}{$\lambda_{or}$} & [7 & 2 & 30 & 22 & 6 & 14 & 18 & 23 & 10 & 15 & 19 & 26 & 3 & 11 & 27 &  \multicolumn{1}{c|}{31,}\\  &13& 9 & 25 & 24 & 8 & 1 & 29 & 21 & 4 & 0 & 28 & 20 & 5 & 12 & 17 & 16] \\ \hline

\multirow{2}{*}{$\lambda_{el}$} & [(13 & 29) & (8 & 24) & (1 & 17) & (9 & 25) & (0 & 16) & (5 & 21) & (4 & 20) & (12 & \multicolumn{1}{c|}{28),}\\ & (6 & 22)  & (2 & 18) & (14 & 30) & (7 & 23) & (15 & 31) & (11 & 27) & (3 & 19) & (10 & 26)] \\ \hline
\multirow{2}{*}{$\lambda_{ol}$} & [(7 & 23) & (2 & 18) & (6 & 22) & (14 & 30) & (10 & 26) & (15 & 31) & (3 & 19) & (11 & \multicolumn{1}{c|}{27),} \\ & (13 & 29), & (9 & 25) & (8 & 24) & (1 & 17) & (4 & 20) & (0 & 16) & (5 & 21) & (12 & 28)] \\ \hline
\end{tabular}
\label{MD_32QAM_Map}
}
\end{table*}

\begin{table*}[h]
\caption{Proposed  $\lambda_{er}$, $\lambda_{or}$, $\lambda_{el}$,  and $\lambda_{ol}$ for $64$-QAM.}
\centering
\resizebox{0.8\columnwidth}{!}{%
\begin{tabular}{|c||cccccccccccccccc|}
\hline
\multirow{4}{*}{$\lambda_{er}$} & [63 & 55 & 61 & 53 & 48 & 56 & 50 & 58 & 62 & 54 & 46 & 38 & 35 & 43 & 51 & \multicolumn{1}{c|}{59,} \\ & 60 & 52 & 36 & 37 & 32 & 33 & 49 & 57 & 47 & 39 & 44 & 45 & 40 & 41 & 34 & \multicolumn{1}{c|}{42,} \\ & 21 & 6 & 4 & 5 & 0 & 1 & 3 & 16 & 14 & 29 & 12 & 13 & 8 & 9 & 24 & \multicolumn{1}{c|}{11,} \\ & 22 & 30 & 28 & 20 & 17 & 25 & 27 & 19 & 23 & 31 & 15 & 7 & 2 & 10 & 26 & \multicolumn{1}{c|}{18]} \\ \hline
\multirow{4}{*}{$\lambda_{or}$} & [0 & 8 & 17 & 2 & 7 & 20 & 13 & 5 & 1 & 9 & 25 & 10 & 15 & 28 & 12 & \multicolumn{1}{c|}{4,} \\ & 16 & 24 & 27 & 26 & 31 & 30 & 29 & 21 & 3 & 11 & 19 & 18 & 23 & 22 & 14 & \multicolumn{1}{c|}{6,} \\& 49 & 57 & 59 & 58 & 63 & 62 & 60 & 52 &
34 & 42 & 51 & 50 & 55 & 54 & 47 & \multicolumn{1}{c|}{39,} \\ & 41 & 33 & 56 & 48 & 53 & 61 & 36 & 44 & 40 & 32 & 43 & 35 & 38 & 46 & 37 & \multicolumn{1}{c|}{45]} \\ \hline

\multirow{4}{*}{$\lambda_{el}$} & [(14 & 46) & (6 & 38) & (4 & 36) & (5 & 37) & (0 & 32) & (1 & 33) & (3 & 35) & (11 & \multicolumn{1}{c|}{43),} \\ & (30 & 62) & (21 & 53) & (12 & 44) & (13 & 45) & (8 & 40) & (9 & 41) & (16 & 48) & (27 & \multicolumn{1}{c|}{59),} \\ & (22 & 54) & (29 & 61) & (20 & 52) & (7 & 39) & (2 & 34) & (17 & 49) & (24 & 56) & (19 &  \multicolumn{1}{c|}{51),} \\ & (23 & 55) & (31 & 63) & (28 & 60) & (15 & 47) & (10 & 42) & (25 & 57) & (26 & 58) & (18 & \multicolumn{1}{c|}{ 50)]} \\ \hline
\multirow{4}{*}{$\lambda_{ol}$} & [(0 & 32) & (8 & 40) & (17 & 49) & (2 & 34) & (7 & 39) & (20 & 52) & (13 & 45) & (5 & \multicolumn{1}{c|}{37),} \\ & (1 & 33) & (16 & 48) & (25 & 57) & (10 & 42) & (15 & 47) & (28 & 60) & (21 & 53) & (4 & \multicolumn{1}{c|}{36),}\\ & (9 & 41) & (24 & 56) & (27 & 59) & (26 & 58) & (31 & 63) & (22 & 54) & (29 & 61) & (12 & \multicolumn{1}{c|}{44),} \\ & (3 & 35) & (11 & 43) & (19 & 51) & (18 & 50) & (23 & 55) & (30 & 62) & (14 & 46) & (6 & \multicolumn{1}{c|}{38)]} \\ \hline 
\end{tabular}
\label{MD_64QAM_Map}
}
\end{table*}

\begin{table*}[h]
\caption{Proposed $\lambda_{er}$, $\lambda_{or}$, $\lambda_{el}$,  and $\lambda_{ol}$ for $128$-QAM.}
\centering
\resizebox{0.8\columnwidth}{!}{%
\begin{tabular}{|c||cccccccccccccccc|}
\hline
\multirow{8}{*}{$\lambda_{er}$} & [100 & 103 & 37 & 86 & 105 & 122 & 32 & 115 & 116 & 118 & 36 & 70 & 104 & 107 & 106 & \multicolumn{1}{c|}{113,} \\ & 102 & 110 & 126 & 111 & 127 & 96 & 114 & 98 & 119 & 101 & 108 & 124 & 109 & 82 & 66 & \multicolumn{1}{c|}{99,} \\& 117 & 94 & 78 & 125 & 64 & 67 & 112 & 97 & 87 & 76 & 79 & 95 & 80 & 72 & 83 & \multicolumn{1}{c|}{74,} \\ & 68 & 71 & 92 & 77 & 88 & 65 & 75 & 90 & 61 & 84 & 69 & 93 & 73 & 89 & 91 & \multicolumn{1}{c|}{120,} \\ & 60 & 31 & 13 & 85 & 24 & 8 & 81 & 121 & 14 & 15 & 12 & 29 & 25 & 11 & 16 & \multicolumn{1}{c|}{3,} \\ & 30 & 5 & 21 & 28 & 9 & 27 & 17 & 19 & 52 & 23 & 7 & 20 & 57 & 56 & 26 & \multicolumn{1}{c|}{1,} \\& 39 & 4 & 22 & 59 & 43 & 10 & 49 & 51 & 38 & 54 & 6 & 41 & 58 & 40 & 42 & \multicolumn{1}{c|}{33,} \\ & 53 & 46 & 47 & 63 & 18 & 2 & 50 & 48 & 62 & 55 & 44 & 45 & 0 & 123 & 34 & \multicolumn{1}{c|}{35]} \\
\hline
\end{tabular}
\label{MD_128QAM_Map}
}
\end{table*}

\begin{table*}[h!]
\caption{Proposed  $\lambda_{er}$, $\lambda_{or}$, $\lambda_{el}$,  and $\lambda_{ol}$ for $256$-QAM.}
\centering
\resizebox{0.87\columnwidth}{!}{%
\begin{tabular}{|c||cccccccccccccccc|}
\hline
\multirow{16}{*}{$\lambda_{er}$} & [206 & 205 & 222 & 237 & 236 & 220 & 252 & 211 & 227 & 195 & 194 & 135 & 215 & 204 & 231 & 199, \\ & 207 & 254 & 221 & 141 & 253 & 200 & 140 & 243 & 167 & 226 & 151 & 210 & 247 & 83 & 230 & 198, \\& 238 & 142 & 201 & 88 & 217 & 157 & 216 & 156 & 131 & 209 & 242 & 166 & 193 & 134 & 213 & 197, \\ & 239 & 255 & 89 & 158 & 173 & 232 & 172 & 188 & 163 & 130 & 183 & 225 & 208 & 150 & 246 & 214, \\& 223 & 218 & 174 & 233 & 249 & 189 & 248 & 136 & 162 & 147 & 241 & 129 & 165 & 81 & 245 & 229, \\ & 202 & 219 & 159 & 191 & 190 & 137 & 153 & 152 & 179 & 145 & 146 & 182 & 224 & 149 & 192 & 133, \\& 203 & 234 & 250 & 138 & 154 & 169 & 185 & 168 & 177 & 161 & 240 & 144 & 181 & 80 & 244 & 212, \\ & 143 & 175 & 251 & 170 & 155 & 187 & 186 & 184 & 178 & 176 & 160 & 128 & 164 & 148 & 132 & 228, \\
\hline
\end{tabular}
\label{MD_256QAM_Map}
}
\end{table*}

\begin{table*}[h!]
\caption{Proposed  $\lambda_{er}$ and $\lambda_{or}$ for $512$-QAM.}
\centering
\resizebox{0.78\columnwidth}{!}{%
\begin{tabular}{|c||cccccccccccccccc|}
\hline
\multirow{32}{*}{$\lambda_{er}$} & [397 & 493 & 479 & 462 & 477 & 456 & 335 & 269 & 162 & 165 & 455 & 502 & 468 & 453 & 439 & 503, \\ & 411 & 334 & 509 & 494 & 461 & 472 & 488 & 285 & 374 & 342 & 372 & 402 & 485 & 501 & 434 & 385, \\& 172 & 440 & 408 & 510 & 495 & 504 & 491 & 473 & 262 & 451 & 498 & 466 & 452 & 406 & 386 & 326, \\ & 475 & 174 & 409 & 511 & 492 & 507 & 271 & 366 & 164 & 481 & 449 & 486 & 471 & 390 & 340 & 260, \\& 463 & 415 & 399 & 446 & 398 & 447 & 410 & 430 & 420 & 422 & 464 & 438 & 388 & 405 & 389 & 484, \\ & 414 & 412 & 396 & 474 & 458 & 442 & 428 & 394 & 423 & 421 & 448 & 384 & 487 & 436 & 437 & 404, \\& 413 & 445 & 459 & 431 & 270 & 395 & 429 & 427 & 416 & 418 & 480 & 432 & 401 & 400 & 496 & 407, \\ & 444 & 506 & 443 & 393 & 489 & 286 & 424 & 426 & 417 & 483 & 261 & 358 & 433 & 403 & 465 & 391, \\
\hline
\end{tabular}
\label{MD_512QAM_Map}
}
\end{table*}

\begin{table*}[h!]
\caption{Proposed $\lambda_{el}$  and $\lambda_{ol}$ for $512$-QAM.}
\centering
\resizebox{0.86\columnwidth}{!}{%
\begin{tabular}{|c||cccccccccccccccc|}
\hline
\multirow{32}{*}{$\lambda_{el}$} & [(172 & 428) & (168 & 424) & (162 & 418) & (165 & 421) & (175 & 431) & (173 & 429) & (164 & 420) & (170 & 426), \\ & (160 & 416) & (167 & 423) & (161 & 417) & (42 & 298) & (43 & 299) & (169 & 425) & (41 & 297) & (57 & 313), \\& (35 & 291) & (33 & 289) & (32 & 288) & (1 & 257) & (64 & 320) & (177 & 433) & (10 & 266) & (58 & 314), \\ & (171 & 427) & (122 & 378) & (59 & 315) & (123 & 379) & (40 & 296) & (121 & 377) & (51 & 307) & (49 & 305), \\& (163 & 419) & (96 & 352) & (65 & 321) & (17 & 273) & (241 & 497) & (80 & 336) & (90 & 346) & (74 & 330), \\ & (106 & 362) & (45 & 301) & (107 & 363) & (56 & 312) & (105 & 361) & (120 & 376) & (34 & 290 & 99 & 355), \\& (97 & 353) & (113 & 369) & (37 & 293) & (39 & 295) & (103 & 359) & (69 & 325) & (124 & 380) & (235 & 491), \\ & (61 & 317) & (11 & 267) & (109 & 365) & (75 & 331) & (104 & 360) & (9 & 265) & (50 & 306) & (98 & 354), \\
\hline
\end{tabular}
\label{MD_512QAM_Map}
}
\end{table*}

\begin{table}[h!]
\caption{Proposed  $\lambda_{er}$ for $1024$-QAM.}
\centering
\resizebox{0.8\columnwidth}{!}{%
\begin{tabular}{|c||cccccccccccccccc|}
\hline
\multirow{64}{*}{$\lambda_{er}$} & [674 & 195 & 226 & 193 & 197 & 229 & 67 & 225 & 129 & 163 & 161 & 165 & 97 & 69 & 203 & 65, \\ & 182 & 54 & 252 & 222 & 220 & 180 & 118 & 176 & 86 & 151 & 148 & 150 & 247 & 246 & 214 & 212, \\& 194 & 199 & 227 & 231 & 131 & 133 & 200 & 99 & 234 & 202 & 37 & 101 & 1 & 33 & 235 & 201, \\ & 158 & 156 & 254 & 22 & 52 & 470 & 23 & 20 & 116 & 183 & 144 & 146 & 215 & 242 & 244 & 213, \\& 192 & 228 & 224 & 162 & 66 & 64 & 3 & 5 & 35 & 451 & 483 & 139 & 485 & 481 & 449 & 233, \\ & 406 & 190 & 404 & 407 & 502 & 223 & 48 & 50 & 468 & 112 & 178 & 87 & 240 & 245 & 210 & 208, \\& 196 & 230 & 130 & 167 & 205 & 237 & 450 & 487 & 453 & 389 & 385 & 77 & 72 & 171 & 137 & 169, \\ & 438 & 188 & 94 & 436 & 400 & 159 & 191 & 500 & 16 & 255 & 114 & 179 & 181 & 84 & 243 & 694, \\
\hline
\end{tabular}
\label{MD_1024QAM_Map}
}
\end{table}

\FloatBarrier
\clearpage
\begin{table}[h!]
\caption{Proposed $\lambda_{or}$ for $1024$-QAM.}
\centering
\resizebox{0.8\columnwidth}{!}{%
\begin{tabular}{|c||cccccccccccccccc|}
\hline
\multirow{64}{*}{$\lambda_{or}$} & [829 & 825 & 827 & 797 & 793 & 893 & 891 & 889 & 856 & 859 & 921 & 861 & 987 & 857 & 1017 & 985, \\ & 910 & 846 & 942 & 908 & 940 & 878 & 782 & 876 & 814 & 783 & 812 & 778 & 808 & 810 & 811 & 809, \\& 831 & 792 & 824 & 795 & 953 & 957 & 955 & 858 & 1021 & 821 & 789 & 1019 & 571 & 569 & 537 & 849, \\ & 774 & 556 & 1006 & 974 & 802 & 911 & 1004 & 844 & 847 & 936 & 874 & 780 & 776 & 815 & 779 & 334, \\& 349 & 826 & 799 & 895 & 863 & 923 & 817 & 573 & 989 & 541 & 785 & 539 & 885 & 881 & 853 & 633, \\ & 558 & 526 & 806 & 524 & 972 & 800 & 775 & 772 & 768 & 904 & 938 & 872 & 879 & 781 & 875 & 813, \\& 828 & 796 & 794 & 890 & 925 & 1016 & 984 & 986 & 538 & 536 & 945 & 851 & 605 & 635 & 603 & 697, \\ & 601 & 590 & 620 & 870 & 838 & 527 & 864 & 804 & 770 & 975 & 771 & 943 & 840 & 842 & 843 & 777, \\
\hline
\end{tabular}
\label{MD_1024QAM_Map}
}
\end{table}

\clearpage
\begin{table}[h!]
\caption{Proposed  $\lambda_{el}$ for $1024$-QAM.}
\centering
\resizebox{0.8\columnwidth}{!}{%
\begin{tabular}{|c||cccccccccccccccc|}
\hline
\multirow{64}{*}{$\lambda_{el}$} & [(195 & 707)) & ((226 & 738) & (193 & 705) & (197 & 709) & (229 & 741) & (67 & 579) & (225 & 737) & (129 & 641), \\ & (163 & 675) & (161 & 673) & (165 & 677) & (97 & 609) & (69 & 581) & (203 & 715) & (65 & 577) & (182 & 694), \\& (54 & 566) & (252 & 764) & (222 & 734) & (220 & 732) & (180 & 692) & (118 & 630) & (176 & 688) & (86 & 598), \\ & (151 & 663) & (148 & 660) & (150 & 662) & (247 & 759) & (246 & 758) & (214 & 726) & (212 & 724) & (194 & 706), \\& (199 & 711) & (227 & 739) & (231 & 743) & (131 & 643) & (133 & 645) & (200 & 712) & (99 & 611) & (234 & 746), \\ & (202 & 714) & (37 & 549) & (101 & 613) & (1 & 513) & (33 & 545) & (235 & 747) & (201 & 713) & (158 & 670), \\& (156 & 668) & (254 & 766) & (22 & 534) & (52 & 564) & (470 & 982) & (23 & 535) & (20 & 532) & (116 & 628), \\ & (183 & 695) & (144 & 656) & (146 & 658) & (215 & 727) & (242 & 754) & (244 & 756) & (213 & 725) & (192 & 704), \\
\hline
\end{tabular}
\label{MD_1024QAM_Map}
}
\end{table}

\clearpage
\begin{table}[h!]
\caption{Proposed  $\lambda_{ol}$ for $1024$-QAM.}
\centering
\resizebox{0.85\columnwidth}{!}{%
\begin{tabular}{|c||cccccccccccccccc|}
\hline
\multirow{64}{*}{$\lambda_{ol}$} & [(349 & 861) & (334 & 846) & (366 & 878) & (345 & 857) & (473 & 985) & (398 & 910) & (335 & 847) & (505 & 1017), \\ & (57 & 569) & (121 & 633) & (91 & 603) & (89 & 601) & (465 & 977) & (249 & 761) & (217 & 729) & (81 & 593), \\& (113 & 625) & (209 & 721) & (208 & 720) & (212 & 724) & (195 & 707) & (196 & 708) & (198 & 710) & (134 & 646), \\ & (102 & 614) & (38 & 550) & (142 & 654) & (390 & 902) & (422 & 934) & (46 & 558) & (396 & 908) & (430 & 942), \\& (377 & 889) & (25 & 537) & (337 & 849) & (185 & 697) & (153 & 665) & (497 & 1009) & (219 & 731) & (49 & 561), \\ & (17 & 529) & (85 & 597) & (83 & 595) & (211 & 723) & (243 & 755) & (210 & 722) & (213 & 725) & (214 & 726), \\& (194 & 706) & (199 & 711) & (230 & 742) & (135 & 647) & (132 & 644) & (166 & 678) & (70 & 582) & (6 & 518), \\ & (236 & 748) & (140 & 652) & (174 & 686) & (110 & 622) & (14 & 526) & (262 & 774) & (294 & 806) & (302 & 814), \\
\hline
\end{tabular}
\label{MD_1024QAM_Map}
}
\end{table}

\section*{Acknowledgment}

\end{document}